\documentclass[12pt]{article}
\usepackage{amssymb}
\usepackage{amsmath}
\usepackage{amsthm}
\usepackage{amsfonts}
\usepackage{setspace}
\usepackage{geometry}
\usepackage{graphicx}
\usepackage{color}

\newtheorem{theorem}{Theorem}
\newtheorem{assumption}{Assumption}
\newtheorem{lemma}{Lemma}

\numberwithin{equation}{section}
\input{tcilatex}
\begin{document}

\title{{\Large Maximum Score Estimation of Preference Parameters for a
Binary Choice Model under Uncertainty} \thanks{%
This work was in part supported by National Science Council of Taiwan
(102-2410-H-001-012), by European Research Council (ERC-2009-StG-240910-
ROMETA) and by the National Research Foundation of Korea Grant funded by the
Korean Government (NRF-2012-S1A3-A2033467). We thank Hidehiko Ichimura,
Liangjun Su, participants at 2013 Asian Meeting of the Econometric Society
and seminar participants at Singapore Management University, a co-editor and
three anonymous referees for helpful comments on this work.}}
\author{Le-Yu Chen \\
Institute of Economics, Academia Sinica \and Sokbae Lee \\
Department of Economics, Seoul National University \and Myung Jae Sung%
\thanks{%
Corresponding author. Address: School of Economics, Hongik University, 94
Wausan-ro, Mapo-Gu, Seoul, South Korea 121-791. E-mail:
mjaesung@hongik.ac.kr.} \\
School of Economics, Hongik University}
\date{1 December 2013}
\maketitle

\begin{abstract}
This paper develops maximum score estimation of preference parameters in the
binary choice model under uncertainty in which the decision rule is affected
by conditional expectations. The preference parameters are estimated in two
stages: we estimate conditional expectations nonparametrically in the first
stage and then the preference parameters in the second stage based on Manski
(1975,\thinspace 1985)'s maximum score estimator using the choice data and
first stage estimates. The paper establishes consistency and derives rate of
convergence of the two-stage maximum score estimator. Moreover, the paper
also provides sufficient conditions under which the two-stage estimator is
asymptotically equivalent in distribution to the corresponding single-stage
estimator that assumes the first stage input is known. These results are of
independent interest for maximum score estimation with nonparametrically
generated regressors. The paper also presents some Monte Carlo simulation
results for finite-sample behavior of the two-stage estimator.%


\end{abstract}


\onehalfspacing

\section{Introduction}

This paper develops a semiparametric two-stage estimator of preference
parameters in the binary choice model where the agent's decision rule is
affected by conditional expectations of outcomes which are uncertain at the
choice-making stage and the preference shocks are nonparametrically
distributed with unknown form of heteroskedasticity. The pioneering papers
of Manski\thinspace (1991,\thinspace 1993) established nonparametric
identification of agents' expectations in the discrete choice model under
uncertainty when the expectations are fulfilled and conditioned only on
observable variables. Utilizing this result, Ahn and Manski\thinspace (1993)
proposed a two-stage estimator for a binary choice model under uncertainty
where agent's utility was linear in parameter and the unobserved preference
shock had a known distribution. Specifically, they estimated the agent's
expectations nonparametrically in the first stage and then the preference
parameters in the second stage by maximum likelihood estimation using the
choice data and the expectation estimates. Ahn\thinspace (1995,\thinspace
1997) extended the two-step approach further. On one hand, Ahn (1995)
considered nonparametric estimation of conditional choice probabilities in
the second stage. On the other hand, Ahn (1997) retained the linear index
structure of the Ahn-Manski model but estimated the preference parameters in
the second stage using average derivative method hence allowing for unknown
distribution of the unobservable. In principle, alternative approaches
accounting for nonparametric unobserved preference shock can also be applied
in the second step estimation of this framework. Well known methods include
Cosslett (1983), Powell et al.\thinspace (1989), Ichimura (1993), Klein and
Spady (1993), and Coppejans (2001), among many others.\medskip

The aforementioned papers allow for nonparametric setting of the
distribution of the preference shock. But the unobserved shock is assumed
either to be independent of or to have specific dependence structure with
the covariates. By contrast, Manski\thinspace (1975,\thinspace 1985)
considered a binary choice model under the conditional median restriction
and thus allowed for general form of heteroskedasticity for the unobserved
shock. It is particularly important, as shown in Brown and Walker (1989), to
account for heteroskedasticity in random utility models. Therefore, this
paper develops the semiparametric two-stage estimation method for the
Ahn-Manski model where the second stage is based on Manski (1975,\thinspace
1985)'s maximum score estimator and thus can accommodate nonparametric
preference shock with unknown form of heteroskedasticity.\medskip

From a methodological perspective, this paper also contributes to the
literature on two-stage M-estimation method with non-smooth criterion
functions. When the true parameter value can be formulated as the unique
root of certain population moment equations, the problem of M-estimation can
be reduced to that of Z-estimation. Chen et al.\thinspace (2003) considered
semiparametric non-smooth Z-estimation problem with estimated nuisance
parameter, while allowing for over-identifying restrictions. Chen and Pouzo
(2009,\thinspace 2012) developed general estimation methods for
semiparametric and nonparametric conditional moment models with possibly
non-smooth generalized residuals. For the general M-estimation problem,
Ichimura and Lee (2010) assumed some degree of second-order expansion of the
underlying objective function and established conditions under which one can
obtain a $\sqrt{N}$-consistent estimator of the finite dimensional parameter
where $N$ is the sample size when the nuisance parameter at the first stage
is estimated at a slower rate. For more recent papers on two-step
semiparametric estimation, see Ackerberg et al. (2012), Chen et al. (2013),
Escanciano et al. (2012, 2013), Hahn and Ridder (2013), and Mammen et al.
(2013), among others. None of the aforementioned papers include the maximum
score estimation in the second stage estimation. \medskip

For this paper, the second stage maximum score estimation problem cannot be
reformulated as a Z-estimation problem. Furthermore, even in the absence of
nuisance parameter, Kim and Pollard (1990) demonstrated that the maximum
score estimator can only have the cube root rate of convergence and its
asymptotic distribution is non-standard. The most closely related paper is
Lee and Pun (2006) who showed that $m$ out of $n$ bootstrapping can be used
to consistently estimate sampling distributions of nonstandard M-estimators
with nuisance parameters. Their general framework includes the maximum score
estimator as a special case, but allowing for only parametric nuisance
parameters. Therefore, established results in the two-stage estimation
literature are not immediately applicable and the asymptotic theory
developed in this paper may also be of independent interest for non-smooth
M-estimation with nonparametrically generated covariates.\medskip

The rest of the paper is organized as follows. Section 2 sets up the binary
choice model under uncertainty and presents the two-stage maximum score
estimation procedure of the preference parameters. Section 3 states
regularity assumptions and derives consistency and rate of convergence of
the estimator. In addition, Section 3 gives conditions under which the
two-stage maximum score estimator is asymptotically equivalent to the
infeasible single-stage maximum score estimator with a known first stage
input. Section 4 presents Monte Carlo studies assessing finite sample
performance of the estimator. Section 5 gives further applications of
maximum score estimation with nonparametrically generated regressors.
Section 6 concludes this paper. Proofs of technical results along with some
preliminary lemmas are given in the Appendices.

\section{Maximum Score Estimation of a Binary Choice Model under Uncertainty}

Suppose an agent must choose between two actions denoted by 0 and 1. The
utility from choosing action $j\in \{0,1\}$ is 
\begin{equation*}
U{}_{j}=v{}_{j}{}^{\prime }\beta {}_{1}+y{}^{\prime }\beta
{}_{2}+\varepsilon {}_{j}.
\end{equation*}%
Realization of the random vector $(v_{j},\varepsilon _{j})\in R^{k}\times R$
is known to the agent before the action is chosen and the random vector $%
y\in R^{p}${\ \ is realized only after the action is chosen. Random vectors $%
(v_{1},\varepsilon _{1})$ and $(v_{0},\varepsilon _{0})$ are not necessarily
identical. Distribution of $y$ depends on the chosen action and realization
of a random vector $x\in R^{q}$. Let $E^{s}(\cdot |\cdot )${\ \ denote the
agent's subjective conditional expectation. Given the realization of $%
(v_{j},\varepsilon _{j})$, the agent chooses the action $d$ that maximizes
the expected utility: 
\begin{equation*}
v_{j}^{\prime }\beta _{1}+E^{s}(y|x,d=j){}^{\prime }\beta _{2}+\varepsilon
_{j},j\in \{0,1\}.
\end{equation*}%
Thus the decision rule has the form 
\begin{equation}
d=1\left\{ z{}^{\prime }\beta
{}_{1}+[E{}^{s}(y|x,d=1)-E{}^{s}(y|x,d=0)]{}^{\prime }\beta
{}_{2}>\varepsilon \right\} ,  \label{e1}
\end{equation}%
where $z\equiv v{}_{1}-v{}_{0},\ \varepsilon \equiv \varepsilon
_{0}-\varepsilon _{1}$, and $1\{\cdot \}$ is an indicator function whose
value is one if the argument is true and zero otherwise. }}\medskip

As in Ahn and Manski (1993), suppose that expectations are fulfilled: 
\begin{equation*}
E{}^{s}(y|x,d=j)=E(y|x,d=j).
\end{equation*}%
We assume that the researcher does not observe realization of $\varepsilon ${%
\ and }$E(y|x,d=j)$, but that of $(z,x,d,y)$.\medskip

Let $G(x)\equiv E(y|x,d=1)-E(y|x,d=0)$ and let $w\equiv (z,G(x))\in \mathcal{%
W}\subset R{}^{k+p}$, where $\mathcal{W}$ denotes the support of the
distribution of $w$. Then, equation (\ref{e1}) can be written as 
\begin{equation}
d=1\{w^{\prime }\beta >\varepsilon \},  \label{eq}
\end{equation}%
where $\beta \equiv (\beta _{1},\beta _{2})$ is a vector of unknown
preference parameters. The set of assumptions leading to the binary choice
model in (\ref{eq}) is equivalent to that of Ahn and Manski (1993, equations
(1)-(3)). Note that $x$ affects the agent's decision only through $G(x)$,
and therefore, $x$ and $z$ can have common elements, as long as the support
of the distribution of $w$ is not contained in any proper linear subspace of 
$R{}^{k+p}$. \medskip

In this paper, we consider an important deviation from Ahn and Manski
(1993)'s setup where the unobserved preference shock $\varepsilon $ is
independent of $(z,x)$ with a known distribution function. Instead, we
consider inference under a flexible specification of the unobserved model
component. Following Manski (1985), we impose the restriction: 
\begin{equation}
\text{Med}(\varepsilon |z,x)=0.  \label{m}
\end{equation}%
The conditional median independence assumption in (\ref{m}) allows for
heteroskedasticity of unknown form, and hence, is substantially weaker than
the assumption imposed in Ahn and Manski (1993). Given (\ref{m}), the model (%
\ref{e1}) then satisfies 
\begin{equation}
\text{Med}(d|z,x)=1\{w^{\prime }\beta >0\}.  \label{e2}
\end{equation}

We may consider sufficient conditions for \eqref{m} in terms of the original
structural errors $\varepsilon _{0}$ and $\varepsilon _{1}$. Recall that $%
\varepsilon \equiv \varepsilon _{0}-\varepsilon _{1}$. Suppose that (i) the
distribution of $(\varepsilon _{0},\varepsilon _{1})$ is the same as that of 
$(\varepsilon _{1},\varepsilon _{0})$ conditional on $x$ and $z$, and (ii)
the support of this common conditional distribution is $R^{2}$. This type of
condition is called conditional exchangeability assumption. Then this
implies that $\varepsilon $ is symmetrically distributed around zero,
thereby implying equation (2.3). For further discussions regarding
conditional exchangeability assumption, see Fox (2007) in the context of
multinomial discrete-choice models and Arellano and Honor\'{e} (2001) for
applications in panel data models, among others. Also, note that the
conditional exchangeability assumption is a sufficient (but not necessary)
condition for equation (\ref{m}). \medskip

{\ Let }$\Theta $ denote the space of preference parameters, and let $%
\Lambda _{j}$, $j\in \{1,...,p\}$, denote the function space of difference
of conditional expectations $E(y_{j}|x,d=1)-E(y_{j}|x,d=0)$. Moreover, let $%
b\equiv (b_{1},b_{2})$ and $\gamma _{j}(x)$, $j\in \{1,...,p\}$, denote
generic elements of $\Theta $ and $\Lambda _{j}$, respectively. Let $\gamma
(x)\equiv (\gamma _{1}(x),...,\gamma _{p}(x))$ and $\Lambda \equiv
\prod\nolimits_{j=1}^{p}\Lambda _{j}$ be the space of $\gamma $. We refer to 
$\beta \equiv (\beta_1,\beta_2) $ and $G(x)$ as the true finite-dimensional
and infinite-dimensional parameters.\medskip

Suppose that data consist of random samples $%
(z{}_{i},x{}_{i},d{}_{i},y{}_{i}),i=1,\cdots ,N$. We estimate in the first
stage the conditional expectations which are not observed. Let $\widehat{G{}}%
(x{}_{i})${\ denote an estimate of the difference in conditional
expectations. Using the estimate }$\widehat{G{}}$, we {estimate the
preference parameters }$\beta $ in the second stage\ by the method of
maximum score estimation of Manski (1975,\thinspace 1985). For any $b$ and $%
\gamma $, define the sample score function%
\begin{equation}
S_{N}(b,\gamma )\equiv \frac{1}{N}\dsum\limits_{i=1}^{N}\tau
_{i}(2d_{i}-1)1\{z_{i}^{\prime }b{}_{1}+\gamma (x_{i}){}^{\prime }b_{2}>0\},
\label{score}
\end{equation}%
where $\tau _{i}\equiv \tau (x_{i})$ is a predetermined weight function to
avoid unduly influences from estimated $G(x_{i})$ at data points carrying
low density. The two-stage estimator of $\beta $ is now defined as 
\begin{equation}
\widehat{\beta }=\arg \max\nolimits_{b\in \Theta }\unit{S}_{N}(b,\widehat{G}%
).  \label{beta_hat}
\end{equation}

We end this section by commenting on inherent features of the maximum score
estimation approach. The zero conditional median assumption does not require
the existence of any error moments and allows heteroskedastic errors of an
unknown form. However, the maximum score approach has its drawbacks, mainly
due to its weak assumption. First, in terms of prediction power, it can
identify unknown parameters up to scale and also only identify whether the
conditional probability of $d=1$ is above or below one half; hence, the
partial effects of covariates are not identified. Second, lack of smoothness
in the objective function makes computation of the estimator difficult and
lets the estimator converge in probability to the true parameter at a rate
of $N^{-1/3}$.

\section{Consistency, Rate of Convergence and Asymptotic Distribution of $%
\widehat{\protect\beta }\label{main results}$}

Let $F(t;b)$ and $f(t;b)$, respectively, denote the distribution and density
of $w^{\prime }b$. To simplify the analysis, we consider fixed trimming such
that $\tau (x)=1(x\in \mathcal{X})$, where $\mathcal{X\subset R}^{q}$ is a
predetermined, compact, and convex subset of the support of $x$. For any
real vector $b$, let $\left\Vert b\right\Vert _{E}$ denote the Euclidean
norm of $b$. For any $p$-dimensional vector of functions $h(x)$, let $%
\left\Vert h\right\Vert _{\infty }\equiv \left\Vert \left( \left\Vert
h_{1}\right\Vert _{\sup },...,\left\Vert h_{p}\right\Vert _{\sup }\right)
\right\Vert _{E}$ where $\left\Vert h_{j}\right\Vert _{\sup }\equiv \sup
\{\left\vert h_{j}(x)\right\vert :x\in \mathcal{X}\}$ and $h_{j}(x)$ denote
the $j$th component of $h$. Let $\widetilde{z}$ be the subvector of $z$
excluding the first component, say $z_{1}$ of $z$. Write $b_{1}=(b_{1,1},%
\widetilde{b}_{1})$ and $\beta _{1}=(\beta _{1,1},\widetilde{\beta }_{1})$.
We assume the following regularity conditions.\medskip

\begin{assumption}
\label{assumption-consist} Assume that:

\begin{itemize}
\item[\textbf{C1}.] $\Theta =\{-1,1\}\times \Upsilon $, where $\Upsilon $ is
a compact subspace of $R^{k+p-1}$and $\left( \widetilde{\beta }_{1},\beta
_{2}\right) $ is an interior point of $\Upsilon $.

\item[\textbf{C2}.] (a) The support of the distribution of $w$ is not
contained in any proper linear subspace of $R{}^{k+p}$. (b) $0<P(d=1|w)<1$
for almost every $w$. (c) For almost every $(\widetilde{z},x)$, the
distribution of $z_{1}$ conditional on $(\widetilde{z},x)$ has everywhere
positive density with respect to Lebesgue measure.

\item[\textbf{C3}.] $\text{Med}(\varepsilon |z,x)=0$ for almost every $(z,x)$%
.

\item[\textbf{C4}.] There is a positive constant $L<\infty $ such that $%
\left\vert F(t_{1};b)-F(t_{2};b)\right\vert \leq L\left\vert
t_{1}-t_{2}\right\vert $ for all $(t_{1},t_{2})\in R^{2}$ uniformly over $%
b\in \Theta $.

\item[\textbf{C5}.] $\left\Vert \widehat{G{}}-G\right\Vert _{\infty
}=o_{p}(1)$.\medskip
\end{itemize}
\end{assumption}

Because the scale of $\beta $ for the model characterized by (\ref{e2})
cannot be identified, Assumption C1 imposes scale normalization by requiring
that the absolute value of the first coefficient is unity. Assumption C2
implies that $F(t;b)$ is absolutely continuous and has density $f(t;b)$ for
each $b\in \{-1,1\}\times \Upsilon $. Assumptions C1 - C3 are standard in
the maximum score estimation literature (see e.g., Manski (1985), Horowitz
(1992), and Florios and Skouras (2008)). Assumption C4 is a mild condition
on the distribution of the index variable $w^{\prime }b$. Assumption C5
requires uniform consistency of the first stage estimation. This assumption
can be easily verified for standard nonparametric estimators such as series
estimators (Newey (1997, Theorem 1)) and the kernel regression estimator
(Bierens (1983, Theorem 1), Bierens (1987, Theorem 2.3.1) and Andrews (1995,
Theorem 1)). \medskip

Given these regularity conditions, we have the following result.\medskip

\begin{theorem}[Consistency]
\label{thm-consist} Let Assumption \ref{assumption-consist} (\textbf{C1} - 
\textbf{C5}) hold. Then the two-stage estimator given by (\ref{beta_hat})
converges to $\beta $ in probability as $N\longrightarrow \infty $.\medskip
\end{theorem}

In addition to consistency, we also study rate of convergence of the
estimator $\widehat{\beta }$. Let $\widetilde{w}\equiv (\widetilde{z},G(x)),$
$\widetilde{b}\equiv (\widetilde{b}_{1},b_{2})$ and $\widetilde{\beta }%
\equiv (\widetilde{\beta }_{1},\beta _{2})$. Let $F_{\varepsilon }(\cdot
|z,x)$ denote the distribution function of $\varepsilon $ conditional on $%
(z,x)$ and $g_{1}(z_{1}|\widetilde{z},x)$ denote the density function of $%
z_{1}$ conditional on $(\widetilde{z},x)$. Let $p_{1}\left( \cdot ,%
\widetilde{z},x\right) $ denote the partial derivative of $P(d=1|z,x)$ with
respect to $z_{1}$. Define the following matrix%
\begin{equation*}
V\equiv \beta _{1,1}E\left[ \tau p_{1}(-\widetilde{w}^{\prime }\widetilde{%
\beta }{}/\beta _{1,1},\widetilde{z},x)g_{1}(-\widetilde{w}^{\prime }%
\widetilde{\beta }/\beta _{1,1}|\widetilde{z},x)\widetilde{w}\widetilde{w}%
^{\prime }\right] .
\end{equation*}

Since the objective function of (\ref{score}) is non-smooth, we require the
nonparametric parameter of the estimation problem should possess certain
degree of smoothness to facilitate derivation of the rate of convergence
result. In particular, we consider the following well known class of smooth
functions (see, e.g., van der Vaart and Wellner (1996, Section 2.7.1)) : for 
$0<\alpha <\infty $, let $C_{M}^{\alpha }$ denote the class of functions $f$:%
$\mathcal{X\longmapsto R}$ with $\left\Vert f\right\Vert _{\alpha }\leq M$
where for any $q$ dimensional vector of non-negative integers $k=\left(
k_{1},...,k_{q}\right) $,%
\begin{equation*}
\left\Vert f\right\Vert _{\alpha }\equiv \max_{\sigma (k)\leq \underline{%
\alpha }}\left\Vert D^{k}f\right\Vert _{\sup }+\max_{\sigma (k)\leq 
\underline{\alpha }}\sup_{x\neq x^{\prime }}\frac{\left\vert
D^{k}f(x)-D^{k}f(x^{\prime })\right\vert }{\left\Vert x-x^{\prime
}\right\Vert _{E}^{\alpha -\underline{\alpha }}}
\end{equation*}%
where $\sigma (k)\equiv \sum\nolimits_{j=1}^{q}k_{q}$, \underline{$\alpha $}
denotes the greatest integer smaller than $\alpha $, and $D^{k}$ is the
differential operator 
\begin{equation*}
D^{k}\equiv \frac{\partial ^{\sigma (k)}}{\partial x_{1}^{k_{1}}\cdot \cdot
\cdot \partial x_{q}^{k_{q}}}.
\end{equation*}%
Given the norm $\left\Vert \cdot \right\Vert _{\alpha }$, for any $p$%
-dimensional vector of functions $h(x)$, let $\left\Vert h\right\Vert
_{\alpha ,p}\equiv \left\Vert \left( \left\Vert h_{1}\right\Vert _{\alpha
},...,\left\Vert h_{p}\right\Vert _{\alpha }\right) \right\Vert _{E}$ where $%
h_{j}(x)$ denote the $j$th component of $h$. Note that $\left\Vert \cdot
\right\Vert _{\alpha ,p}$ is a stronger norm than $\left\Vert \cdot
\right\Vert _{\infty }$ used in condition C5 for the uniform consistency of
the first stage estimator. \medskip

The regularity conditions imposed for the convergence rate result are stated
as follows.\medskip

\begin{assumption}
\label{assumption-rate}Assume that:

\begin{itemize}
\item[\textbf{C6.}] The support of $\widetilde{z}$ is bounded.

\item[\textbf{C7}.] There is a positive constant $\overline{B}<\infty $ such
that (i) for every $z_{1}$ and for almost every $\left( \widetilde{z}%
,x\right) $, 
\begin{equation*}
g_{1}(z_{1}|\widetilde{z},x)<\overline{B},\left\vert \partial g_{1}(z_{1}|%
\widetilde{z},x)/\partial z_{1}\right\vert <\overline{B},\text{ and }%
\left\vert \partial ^{2}g_{1}(z_{1}|\widetilde{z},x)/\partial
z_{1}^{2}\right\vert <\overline{B},
\end{equation*}%
and (ii) for non-negative integers $i$ and $j$ satisfying $i+j\leq 2$, 
\begin{equation*}
\left\vert \partial ^{i+j}F_{\varepsilon }(t|z,x)/\partial t^{i}\partial
z_{1}^{j}\right\vert <\overline{B}
\end{equation*}%
for every $t$ and $z_{1}$ and for almost every $\left( \widetilde{z}%
,x\right) $.

\item[\textbf{C8}.] All elements of the vector $\widetilde{w}$ have finite
third absolute moments.

\item[\textbf{C9}.] The matrix $V$ is positive definite.

\item[\textbf{C10}.] For each $j\in \{1,...,p\}$, $\Lambda
_{j}=C_{M}^{\alpha }$ for some $\alpha >q$ and $M<\infty $.

\item[\textbf{C11}.] $\left\Vert \widehat{G{}}-G\right\Vert _{\alpha
,p}=O_{p}(\varepsilon _{N})$ where $\varepsilon _{N}$\ is a non-stochastic
positive real sequence such that $N^{1/3}\varepsilon _{N}\leq 1$ for each $N$%
.\medskip
\end{itemize}
\end{assumption}

Assumption C6 is standard in deriving asymptotic properties of Manski's
maximum score estimator (see, e.g. Kim and Pollard (1990), pp.\thinspace 213
- 216). Assumption C7 requires some smoothness of the density $g_{1}(z_{1}|%
\widetilde{z},x)$ and the distribution $F_{\varepsilon }(t|z,x)$. Assumption
C8 is mild. Since $-V$ corresponds to the second order derivative of $%
E[S_{N}(b,\gamma )]$ with respect to $\widetilde{b}$ evaluated at true
parameter values, Assumption C9 is analogous to the classic condition of
Hessian matrix being non-singular in the M-estimation framework. Assumption
C10 imposes smoothness for the nonparametric parameter $\gamma $ and hence
helps to control complexity of the space $\Lambda $. The requirement $\alpha
>q$ is in line with the literature of two-stage semiparametric estimation
with non-smooth objective functions (See, e.g., Chen et.\thinspace
al.\thinspace (2003, Example 2, pp.\thinspace 1601-1603) and Ichimura and
Lee (2010, Section 4.1, pp.\thinspace 258-259)). \medskip

Assumption C11 requires that the first stage estimator should converge under
the norm $\left\Vert \cdot \right\Vert _{\alpha ,p}$ at a rate no slower
than $N^{-1/3}$. Note that convergence of $\widehat{G{}}$ to $G$ in the norm 
$\left\Vert \cdot \right\Vert _{\alpha ,p}$ also implies uniform convergence
of derivatives of $\widehat{G{}}$ to those of $G$. For integer-valued $%
\alpha >0$, Assumption C11 is fulfilled provided that for vector of
non-negative integers $k=\left( k_{1},...,k_{q}\right) $ that satisfies $%
\sigma (k)\leq \alpha $, 
\begin{equation}
\left\Vert D^{k}\widehat{G}_{t,j}-D^{k}G_{t,j}\right\Vert _{\sup
}=O_{p}(\varepsilon _{N})  \label{der}
\end{equation}%
where $\widehat{G}_{t,j}(x)$ denotes the estimate of $G_{t,j}(x)\equiv
E(y_{j}|x,d=t)$ for $\left( t,j\right) \in \{0,1\}\times \{1,...,p\}$. The
condition (\ref{der}) can also be verified for series estimators (Newey
(1997, Theorem 1)) and the kernel regression estimator (Andrews (1995,
Theorem 1)).\medskip

\begin{theorem}[Rate of Convergence]
\label{thm-rate}In addition to Assumption \ref{assumption-consist} (\textbf{%
C1} - \textbf{C5}), let Assumption \ref{assumption-rate} (\textbf{C6} - 
\textbf{C11}) also hold. Then $\left\Vert \widehat{\beta }-\beta \right\Vert
_{E}=O_{p}(N^{-1/3})$.\medskip
\end{theorem}

If $G$ were priorly known to the researcher, the preference parameters $%
\beta $ could be estimated by the single stage maximum score estimator $%
\widehat{\beta }_{G}$, defined as 
\begin{equation}
\widehat{\beta }_{G}=\arg \max\nolimits_{b\in \Theta }\unit{S}_{N}(b,G).
\label{beta_hat_G}
\end{equation}%
Kim and Pollard\thinspace (1990) showed that $\widehat{\beta }_{G}$
converges to $\beta $ at cube root rate and established its asymptotic
distribution. In the case of unknown $G$, Theorem \ref{thm-rate} implies
that the two-stage estimator $\widehat{\beta }$ retains the same convergence
rate as the infeasible estimator $\widehat{\beta }_{G}$. Indeed if condition
C11 is strengthened for faster convergence of first stage estimates, we can
establish the oracle property that $N^{1/3}(\widehat{\beta }-\beta )$ and $%
N^{1/3}(\widehat{\beta }_{G}-\beta )$ have the same limiting distribution.
Therefore, the inference on $\beta $ can be carried out by subsampling
(Delgado et al.\thinspace (2001)) since the standard bootstrap cannot be
used to estimate the distribution of the maximum score estimator
consistently (Abrevaya and Huang (2005)). We now state the asymptotic
distributional equivalence result in the next theorem.\medskip

\begin{theorem}[Asymptotic Distribution]
\label{thm-distribution}Suppose all assumptions stated in Theorem \ref%
{thm-rate} hold with the additional restriction that the sequence $%
\varepsilon _{N}$ stated in \textbf{C11 }further satisfies $\varepsilon
_{N}=o(N^{-1/3})$. Then $N^{1/3}(\widehat{\beta }-\beta )$ is asymptotically
equivalent in distribution to $N^{1/3}(\widehat{\beta }_{G}-\beta )$.
\end{theorem}

\section{Monte Carlo Simulations}

We employ the following data generating process (DGP) in simulation study of
the two-stage maximum score estimator:%
\begin{equation*}
d=1\{z\beta _{1}+G(x)\beta {}_{2}>\varepsilon \},
\end{equation*}%
where $G(x)=E(y|x,d=1)-E(y|x,d=0)$, $z\sim \text{Logistic}$, $x\sim N(0,1)$
and $\varepsilon =0.25\eta \sqrt{1+z^{2}+x^{2}}$ with $\eta |(x,z)\sim
N(0,1) $. The scalar random variable $y$ is generated according to 
\begin{equation}
y=d(\gamma _{01}+\gamma _{11}m(x)+u_{1})+(1-d)(\gamma _{00}+\gamma
_{10}m(x)+u_{0}),  \label{y}
\end{equation}%
where $(u_{1},u_{0})$ are independent of $(x,z,\varepsilon )$ and are
jointly normally distributed with $E(u_{1})=E(u_{0})=0$, $%
Var(u_{1})=Var(u_{0})=\sigma _{u}^{2}$, and $Cov(u_{1},u_{0})=\rho $. Given (%
{\ref{y}), \ }%
\begin{equation*}
G(x)=\gamma _{01}-\gamma _{00}+(\gamma _{11}-\gamma _{10})m(x).
\end{equation*}%
We consider the following two types of the $m(x)$ function:%
\begin{align}
\text{Linear} &\text{ : } m(x)=x,  \label{L} \\
\text{Nonlinear} &\text{ : } m(x)=x^{2}\tan ^{-1}x.  \label{N}
\end{align}%
The true parameter values are specified as follows: $\beta_1 = 1$, $\beta_2
= 1$, $\gamma_{01} = 0.2$, $\gamma_{11} = 0.1$, $\gamma_{00} = 0.1$, $%
\gamma_{10} = 0.4$, $\rho = -0.8$, and $\sigma_u = 0.33$.


We compare infeasible single-stage estimator using $(z,G(x))$ as regressors
and also the feasible two-stage estimator using $(z,\widehat{G}(x))$ as
regressors. We consider both parametric and nonparametric first stage
estimators. For the former, we estimate $E(y|x,d=j)$ by running OLS of $y$
on $x$ with an intercept term using $d=j$ subsamples. For the latter, we
implement Nadaraya-Watson kernel regression estimators. The nonparametric
estimators of $E(y|x,d=j)$, $j\in \{0,1\}$ are constructed as 
\begin{equation}
\frac{\sum\limits_{i=1}^{N}y_{i}K(\widehat{\sigma }_{j}^{-1}h_{N}^{-1}\left(
x-x_{i}\right) )1\{d_{i}=j\}}{\sum\limits_{i=1}^{N}K(\widehat{\sigma }%
_{j}^{-1}h_{N}^{-1}\left( x-x_{i}\right) )1\{d_{i}=j\}}  \label{nwk}
\end{equation}%
where $\widehat{\sigma }_{j}$ is the estimated standard deviation of $x_{i}$
conditional on $d_{i}=j$, $K(.)$ is a univariate kernel function and $h_{N}$
is a deterministic bandwidth sequence. We use two types of kernel and
bandwidth configurations. \medskip

For the first type, we use the second-order Gaussian kernel and set $h_{N}$
to be $cN^{-1/5}$ for various values of the bandwidth scale $c$. For the
second type, we use the following 8th order kernel function (see, e.g.,
Bierens (1987, p. 112) and Andrews (1995, p.\thinspace 567)): 
\begin{equation}
K(x)\equiv \dsum\limits_{s=1}^{4}a_{s}\left\vert b_{s}^{{}}\right\vert
^{-1}\exp \left[ -x^{2}/(2b_{s}^{2})\right] ,  \label{K}
\end{equation}%
where the constants $\left( a_{s},b_{s}\right) ,$ $s\in \{1,...,4\}$ satisfy 
\begin{equation}
\dsum\limits_{s=1}^{4}a_{s}=1\text{ and }\dsum%
\limits_{s=1}^{4}a_{s}b_{s}^{2l}=0\text{ for }l\in \{1,2,3\}\text{.}
\label{ab}
\end{equation}%
We specify $b_{s}=s^{-1/2}$ and then solve $a_{s}$ as solution of the system
of linear equations (\ref{ab}). Associated with this kernel, the bandwidth $%
h_{N}$ is set to be $cN^{-19/360}$ for various values of the scale $c$.%
\footnote{%
As noted by Bierens (1987, p.\thinspace 113), choice of the constants $%
\left( a_{s},b_{s}\right) $ for the kernel function is less crucial since
its effect on asymptotic variance of the conditional mean estimator can be
captured via the bandwidth scale $c$.} By Theorem 1(b) of Andrews (1995),
kernel regression estimator of $G(x)$ based on the second type configuration
has convergence property required in (\ref{der}) with $\sigma (k)\leq 2$ and 
$\varepsilon _{N}=N^{-41/120}$, thus fulfilling regularity conditions C5 and
C11 of Section \ref{main results}. The first stage estimation with the
second-order kernels satisfies condition C5 but may not satisfy C11;
however, we experiment with the second-order kernels as well since kernel
estimates with the second-order kernels often outperform those with the
higher-order kernels in small samples.\footnote{%
See e.g., Marron and Wand (1992) and Efromovich (2001) for theoretical
arguments why the higher-order kernels may perform poorly in small samples.}%
\medskip

To implement the second-stage estimator using nonparametric first stage
estimators, we trim the data by setting $\tau _{i}=1\{\left\vert
x_{i}\right\vert \leq 1.95\}$ where $\tau _{i}$ is the weight introduced in (%
\ref{score}). The estimates of $\beta _{1}$ and $\beta _{2}$ are obtained
using grid search method.{\ }We report simulation results of $\widehat{\beta 
}_{2}$ for the parameter capturing the agent's uncertainty. Let $\widehat{%
\beta }_{2,\func{Single}}$, $\widehat{\beta }_{2,OLS}$, $\widehat{\beta }%
_{2,Kernel\_2nd}$ and $\widehat{\beta }_{2,Kernel\_8th}$ respectively denote
the estimators $\widehat{\beta }_{2}$ that are constructed based on the
infeasible single-stage, two-stage (OLS first stage) and two-stage (kernel
regression first stage implemented with the 2nd and 8th order kernels)
maximum score estimators. We compute bias, median, root mean squared error
(RMSE), mean absolute deviation (mean AD) and median absolute deviation
(median AD) of these estimators based on 1000 simulation repetitions for
sample size $N\in \{300,500,1000\}$.\medskip

Tables 1-6 present simulation results for the four types of estimators of $%
\beta _{2}$ under linear and nonlinear designs of the $G(x)$ function.
Tables 7 and 8 graph the simulated empirical distribution functions (edf)
for $N^{1/3}(\widehat{\beta }_{2,\func{Single}}-\beta _{2})$, $N^{1/3}(%
\widehat{\beta }_{2,OLS}-\beta _{2})$, $N^{1/3}(\widehat{\beta }%
_{2,Kernel\_2nd}-\beta _{2})$ and $N^{1/3}(\widehat{\beta }%
_{2,Kernel\_8th}-\beta _{2})$. As expected, for linear setup of $G$ the
estimator $\widehat{\beta }_{2,OLS}$ enjoys the best overall finite-sample
performance among all two-stage estimators. However, this estimator also
incurs huge bias when agent's conditional expectation is nonlinear. For the
estimators $\widehat{\beta }_{2,Kernel\_2nd}$ and $\widehat{\beta }%
_{2,Kernel\_8th}$, the function $G$ is nonparametrically estimated at the
first stage. Hence regardless of nonlinearity of $G$, we see that the
simulated bias, RMSE, mean AD and median AD of these estimators generally
decrease as sample size grows. \medskip

We note that the edf curves of Tables 7 and 8 for the (kernel first-stage)
two-stage estimators broadly match shapes of those for the infeasible
estimators. Interestingly, finite sample behavior of the estimator $\widehat{%
\beta }_{2,Kernel\_2nd}$ fits that of $\widehat{\beta }_{2,\func{Single}}$
better than its counterpart implemented with the 8th order kernel. Use of
higher order kernels allows for verification of convergence of $\widehat{G}$
to $G$ in the strong norm $\left\Vert \cdot \right\Vert _{\alpha ,p}$.
However, as well known in the literature, the estimates with the
higher-order kernels seem to perform poorly in simulations relative to those
with the second-order kernels. The superb performance of $\widehat{\beta }%
_{2,Kernel\_2nd}$ suggests that the asymptotic distributional equivalence
result in Theorem 3 may not give us sharp asymptotics and there is scope to
develop further asymptotic theory. This is an interesting future research
topic.

\section{Further Applications of Two-Step Maximum Score Estimation with
First-Stage Nonparametric Estimation}

Our paper has been motivated by estimation problem in the binary choice
model under uncertainty. However, the resulting estimator has wider
applicability than just this model. To further motivate our two-step
estimation procedure, this section gives a couple of additional econometric
models for which unknown parameters can be estimated by maximum score with
nonparametrically generated regressors.\medskip

We first consider maximum score estimation of an incomplete information
games. Aradillas-Lopez (2012) developed a two-step procedure for estimation
of incomplete information games with Nash equilibrium behavior. Equation (2)
of Aradillas-Lopez (2012, p.\thinspace 123) gives a description of players'
behavior in a $2\times 2$ game: 
\begin{align*}
Y_{1}& =1\{X_{1}^{\prime }\beta _{1}+\Delta _{1}\text{Pr}[Y_{2}=1|X]-\zeta
_{1}\geq 0\}, \\
Y_{2}& =1\{X_{2}^{\prime }\beta _{2}+\Delta _{2}\text{Pr}[Y_{1}=1|X]-\zeta
_{2}\geq 0\},
\end{align*}%
where $Y_{p}\in \{0,1\}$ is the binary action for player $p=1,2$, $X_{p}$
and $\zeta _{p}$ are observable and unobservable payoff covariates, $X\equiv
(X_{1}^{\prime },X_{2}^{\prime })^{\prime }$, and $\{(\beta _{p},\Delta
_{p}):p=1,2\}$ are unknown parameters.\medskip

Aradillas-Lopez (2012, Assumption A0, p.\thinspace 122) assumed that
players' behavior corresponds to a Bayesian-Nash equilibrium with a
degenerate selection mechanism. He further assumed that $\zeta _{1}$ and $%
\zeta _{2}$ are independent of each other, independent of $X$, and of the
selection mechanism.\medskip

We can make the same assumptions as in Aradillas-Lopez (2012), with one
exception. As in the previous section, we consider $\text{Med}(\zeta
_{p}|X)=0$ almost surely, instead of assuming the full independence between $%
\zeta _{p}$ and $X$, where $p=1,2$. Allowing for dependence between $\zeta
_{p}$ and $X$ might be important in applications when we suspect possible
interactions between observed covariates and unobserved components that
affect players' payoffs. Then for each $p=1,2$, we can estimate $(\beta
_{p},\Delta _{p})$ by running maximum score regression of $Y_{p}$ on $X_{p}$
and $G_{-p}(X)\equiv \text{Pr}[Y_{-p}=1|X]$ with the nonparametric first
stage estimation of $G_{-p}(X)$. Therefore, methodology of the present paper
can be applied to extension of Aradillas-Lopez (2012)'s context allowing
unobserved payoffs to exhibit unknown form of heteroskedasticity. \medskip

Our second application, which is based on Fox (2007), is maximum score
estimation of multinomial discrete-choice models using a subset of choices
under endogeneity. Fox (2007) proposed pairwise maximum score estimation of
multinomial discrete-choice models using a subset of choices. For
simplicity, assume that a researcher has data on only two choice, say $1$
and $2$, among $J(\geq 3)$ alternatives, and also assume that there exists
an endogenous covariate. Fox (2007, p.\thinspace 1013) solved the
endogeneity problem by including, instead of the endogenous covariate,
fitted values from the OLS regression of the endogenous covariate, say
price, on a vector of instruments. We can extend Fox (2007) to allow for
nonparametric fitted values. Then this extension again can be accommodated
in the framework of maximum score estimation with nonparametrically
generated regressors.

\section{Conclusions}

This paper has developed maximum score estimation of preference parameters
in the binary choice model under uncertainty in which the decision rule is
affected by conditional expectations. The estimation procedure is
implemented in two stages: we estimate conditional expectations
nonparametrically in the first stage and obtain the maximum score estimate
of the preference parameters in the second stage using choice data and the
first stage estimates. The paper has shown consistency and convergence rate
of the two-stage maximum score estimator. Moreover, we also establish the
oracle property in terms of asymptotic equivalence in distribution of the
two-stage estimator and its corresponding infeasible single-stage version.
These results are of independent interest for maximum score estimation with
nonparametrically generated regressors.\medskip

It would be an alternative approach to develop the second stage estimator
using Horowitz (1992)'s smoothed maximum score estimator or using a Laplace
estimator proposed in Jun, Pinkse, and Wan (2013). These alternative methods
would produce faster convergence rates but require extra tuning parameters.
Alternatively, we might build the second stage estimator based on Lewbel
(2000), who introduced the idea of a special regressor satisfying certain
conditional independence restriction. These are interesting future research
topics.

\appendix%

\section{Proof of Consistency}

Recall that $w=(z,G(x))$ and $S_{N}(b,\gamma )$ is the sample score function
defined by (\ref{score}). We first state and prove a preliminary lemma that
will be invoked in proving Theorem \ref{thm-consist} of the paper.\medskip

\begin{lemma}
\label{L1} Under Assumptions C1, C4 and C5, 
\begin{equation}
\sup\limits_{b\in \Theta }\left\vert S_{N}(b,\widehat{G})-S_{N}(b,G)\right%
\vert \overset{p}{\longrightarrow }0.  \label{s1}
\end{equation}
\end{lemma}

\begin{proof}[Proof of Lemma \protect\ref{L1}]
Note that
\begin{equation}
\left\vert S_{N}(b,\widehat{G})-S_{N}(b,G)\right\vert \leq \frac{1}{N}%
\dsum\limits_{i=1}^{N}\tau _{i}1\left\{ \left\vert ({\widehat{G}(x}_{i})-G{(x%
}_{i}))^{\prime }b_{2}\right\vert \geq \left\vert w_{i}^{\prime
}b\right\vert \right\} .  \label{s}
\end{equation}%
By Assumption C1, $\left\Vert b_{2}\right\Vert _{E}<B_{2}$ for some finite
positive constant $B_{2}$. Therefore, the right-hand side of the inequality (%
\ref{s}) is bounded above by%
\begin{equation}
\tilde{\Gamma} _{N}\equiv P_{N}\left( \tau =1,B_{2}\left\Vert \widehat{G{}}%
-G\right\Vert _{\infty }\geq \left\vert w^{\prime }b\right\vert \right),
\label{b}
\end{equation}%
where $P_{N}$ denotes the empirical probability. Note that the term (\ref{b}%
) is further bounded above by%
\begin{equation}
\Gamma _{N}\equiv P_{N}\left( B_{2}\left\Vert \widehat{G{}}-G\right\Vert
_{\infty }\geq \left\vert w^{\prime }b\right\vert \right) .  \label{gamma}
\end{equation}%
Let $E_{\eta }$ denote the event $\left\Vert \widehat{G{}}-G\right\Vert
_{\infty }<\eta $ for some $\eta >0$. Then given $\epsilon >0$,%
\begin{eqnarray*}
&&P(\sup\nolimits_{b\in \Theta }\Gamma _{N}>\epsilon )\leq
P(\sup\nolimits_{b\in \Theta }\Gamma _{N}>\epsilon ,E_{\eta })+P(E^c_{\eta })
\\
&\leq &P\left[ \sup\nolimits_{b\in \Theta }P_{N}\left( B_{2}\eta \geq
\left\vert w^{\prime }b\right\vert \right) >\epsilon \right] +P(E^c_{\eta }).
\end{eqnarray*}%
By Assumption C5, $P(E^c_{\eta })\longrightarrow 0$ as $N\longrightarrow
\infty $. Hence, to show (\ref{s1}), it remains to establish that as $%
N\longrightarrow \infty $,%
\begin{equation}
P\left[\sup\nolimits_{b\in \Theta }P_{N}\left( B_{2}\eta \geq \left\vert
w^{\prime }b\right\vert \right) >\epsilon \right] \longrightarrow 0.
\label{delta}
\end{equation}%
Note that by Assumption C4, $P\left( B_{2}\eta \geq \left\vert w^{\prime
}b\right\vert \right) \leq 2LB_{2}\eta $. Therefore, we have that
\begin{eqnarray}
&&P \left[ \sup\nolimits_{b\in \Theta }P_{N}\left( B_{2}\eta \geq \left\vert
w^{\prime }b\right\vert \right) >\epsilon \right]  \notag \\
&\leq &P \left[ \sup\nolimits_{b\in \Theta }\left\vert P_{N}\left( B_{2}\eta
\geq \left\vert w^{\prime }b\right\vert \right) -P\left( B_{2}\eta \geq
\left\vert w^{\prime }b\right\vert \right) \right\vert >\epsilon
-2LB_{2}\eta \right],  \label{s2}
\end{eqnarray}%
where $\eta$ is taken to be sufficiently small such that $\epsilon
-2LB_{2}\eta > 0$ for the given $\epsilon$. By Lemma 9.6, 9.7 (ii) and 9.12
(i) of Kosorok (2008), the family of sets $\left\{ B_{2}\eta \geq \left\vert
w^{\prime }b\right\vert \right\} $ for $b\in \Theta $ forms a Vapnik-\v{C}%
ervonenkis class. Therefore, by Glivenko-Cantelli Theorem (see, e.g. Theorem
2.4.3 of van der Vaart and Wellner (1996)), the right-hand side of (\ref{s2}%
) tends to zero as $N\longrightarrow \infty $. Hence, the convergence result
in (\ref{delta}) holds and Lemma \ref{L1} thus follows.
\end{proof}

We now prove Theorem \ref{thm-consist} for consistency of $\widehat{\beta }$%
.\medskip

\begin{proof}[Proof of Theorem \protect\ref{thm-consist}]
For any $(b,\gamma )$, define
\begin{equation*}
S(b,\gamma )\equiv E\left[ \tau (2d-1)1\{z{}^{\prime }b{}_{1}+\gamma
(x)^{\prime }b{}_{2}>0\}\right] .
\end{equation*}%
Given Assumptions C1 - C3 and by Manski (1985, Lemma 3, p.\thinspace 321), $%
\beta $ uniquely satisfies $\beta =\arg \max_{b\in \Theta }S(b,G)$. We now
look at the difference%
\begin{equation}
\left\vert S_{N}(b,\widehat{G})-S(b,G)\right\vert \leq \left\vert S_{N}(b,%
\widehat{G})-S_{N}(b,G)\right\vert +\left\vert S_{N}(b,G)-S(b,G)\right\vert ,
\label{d}
\end{equation}%
where by Lemma \ref{L1}, the first term of the right-hand side of (\ref{d})
converges to zero in probability uniformly over $b\in \Theta $, whilst by
Manski (1985, Lemma 4, p.\thinspace 321), the second term converges to zero
almost surely uniformly over $b\in \Theta $. Therefore, we have that
\begin{equation*}
\sup\limits_{b\in \Theta }\left\vert S_{N}(b,\widehat{G})-S(b,G)\right\vert
\overset{p}{\longrightarrow }0.
\end{equation*}%
By Lemma 5 of Manski (1985, p.\thinspace 322), $S(b,G)$ is continuous in $b$%
. Given these results, Theorem 1 thus follows by application of the
consistency theorem in Newey and McFadden (1994, Theorem 2.1).
\end{proof}

\section{Lemma on the Rates of Convergence of a Two-Stage M-Estimator with a
Non-smooth Criterion Function}

We first present and prove a general lemma establishing the rates of
convergence of a general two-stage M-estimator under high level assumptions.
In next section, we prove Theorem \ref{assumption-rate} by verifying these
assumptions for the particular estimator given by (\ref{beta_hat}) under the
regularity conditions of C1 - C11.

To present a general result, let $s\mapsto m_{\theta ,h}(s)$ be measurable
functions indexed by parameters $\left( \theta ,h\right) $. Let $\Theta $
and $H$ be the space of parameters $\theta $ and $h$, respectively. Let $%
\left( \theta ^{\ast },h^{\ast }\right) $ denote the true parameter value.
We assume $\left( \theta ^{\ast },h^{\ast }\right) \in \Theta \times H$. Let 
$S_{N}\left( \theta ,h\right) \equiv \sum_{i=1}^{N}m_{\theta ,h}(s_{i})/N$
be the empirical criterion of the M-estimation problem where $\left(
s_{i}\right) _{i=1}^{N}$ are i.i.d. random vectors. Suppressing the
individual index, let $S\left( \theta ,h\right) \equiv E\left[ m_{\theta
,h}(s)\right] $ be the population criterion. For a given first stage
estimate $\widehat{h}$, let the estimator $\widehat{\theta }$ be constructed
as 
\begin{equation}
\widehat{\theta }=\arg \sup\nolimits_{\theta \in \Theta }S_{N}\left( \theta ,%
\widehat{h}\right) \text{.}  \label{thetahat}
\end{equation}

Let $d_{\Theta }(\theta ,\theta ^{\ast })$ and $d_{H}(h,h^{\ast })$ be
non-negative functions measuring discrepancies between $\theta $ and $\theta
^{\ast }$, and $h$ and $h^{\ast }$, respectively. Note that $d_{\Theta }$
and $d_{H}$ are usually related to but not necessarily the same as the
metrics specified for the spaces $\Theta $ and $H$. Given a non-stochastic
positive real sequence $\varepsilon _{N}$, define $H_{N}(C)\equiv \{h\in
H:d_{H}(h,h^{\ast })\leq C\varepsilon _{N}\}$. To simplify the presentation,
we use the notation $\lesssim $ to denote being bounded above up to a
universal constant. Define the recentered criterion 
\begin{equation}
\widetilde{S}_{N}(\theta ,h)\equiv (S_{N}(\theta ,h)-S_{N}(\theta ^{\ast
},h))-(S(\theta ,h)-S(\theta ^{\ast },h)).  \label{recentered criterion}
\end{equation}%
The following lemma modifies the rate of convergence results developed by
van der Vaart (1998, Theorem 5.55) and provides sufficient conditions
ensuring that $\widehat{\theta }$ retains the same convergence rate as it
would have if $h^{\ast }$ were known.\medskip

\begin{lemma}[Rate of convergence for a general two-stage M-estimator]
\label{rate of convergence lemma} For any fixed and sufficiently large $C>0$%
, assume that for all sufficiently large $N$, 
\begin{equation}
\sup\nolimits_{h\in H_{N}(C)}\left\vert S(\theta ^{\ast },h)-S(\theta ^{\ast
},h^{\ast })\right\vert \lesssim (C\varepsilon _{N})^{2}
\label{inequality 2}
\end{equation}%
and there is a sequence of non-stochastic functions $e_{N}:\Theta \times
H_{N}(C)\longmapsto R$ such that for all sufficiently small $\delta >0$ and
for every $\left( \theta ,h\right) \in \Theta \times H_{N}(C)$ satisfying $%
d_{\Theta }(\theta ,\theta ^{\ast })\leq \delta $, 
\begin{eqnarray}
&&S(\theta ,h)-S(\theta ^{\ast },h^{\ast })+e_{N}(\theta ,h)\lesssim
-d_{\Theta }^{2}(\theta ,\theta ^{\ast })+d_{H}^{2}(h,h^{\ast }),
\label{quadratic expansion} \\
&&\sup_{d_{\Theta }(\theta ,\theta ^{\ast })\leq \delta ,\left( \theta
,h\right) \in \Theta \times H_{N}(C)}\left\vert e_{N}(\theta ,h)\right\vert
\lesssim C\delta \varepsilon _{N},  \label{inequality 1}
\end{eqnarray}%
and 
\begin{equation}
E\left[ \sup_{d_{\Theta }(\theta ,\theta ^{\ast })\leq \delta ,\left( \theta
,h\right) \in \Theta \times H_{N}(C)}\left\vert \widetilde{S}_{N}(\theta
,h)\right\vert \right] \lesssim \frac{\phi _{N}(\delta )}{\sqrt{N}},
\label{maximal inequality}
\end{equation}%
where $\phi _{N}(\delta )$ is a sequence of functions defined on $\left(
0,\infty \right) $ and satisfies that $\phi _{N}(\delta )\delta ^{-\alpha }$
is decreasing for some $\alpha <2$. Suppose $d_{H}(\widehat{h},h^{\ast
})=O_{p}(\varepsilon _{N})$, $d_{\Theta }(\widehat{\theta },\theta ^{\ast
})=o_{p}(1)$ and there is a non-stochastic positive real sequence $\delta
_{N}$ which tends to zero as $N\longrightarrow \infty $ and satisfies that $%
\varepsilon _{N}\leq \delta _{N}$ and $\phi _{N}(\delta _{N})\leq \sqrt{N}%
\delta _{N}^{2}$ for every $N$. Then $d_{\Theta }(\widehat{\theta },\theta
^{\ast })=O_{p}(\delta _{N})$.
\end{lemma}

\begin{proof}
Based on the peeling technique of van der Vaart (1998, Theorem 5.55), for
each natural number $N$, integer $j$ and positive real $M$, construct the
set
\begin{equation*}
A_{N,j,M}(C)\equiv \big\{(\theta ,h)\in \Theta \times H_{N}(C):2^{j-1}\delta
_{N}<d_{\Theta }(\theta ,\theta ^{\ast })\leq 2^{j}\delta
_{N},\,d_{H}(h,h^{\ast })\leq 2^{-M}d_{\Theta }(\theta ,\theta ^{\ast })%
\big\}.
\end{equation*}%
Then we have that for any $\epsilon >0$,%
\begin{eqnarray}
&&P\left( d_{\Theta }(\widehat{\theta },\theta ^{\ast })\geq 2^{M}\left(
\delta _{N}+d_{H}(\widehat{h},h^{\ast })\right) ,\widehat{h}\in
H_{N}(C)\right)  \notag \\
&\leq &P(2d_{\Theta }(\widehat{\theta },\theta ^{\ast })>\epsilon )+P\left( (%
\widehat{\theta },\widehat{h})\in \dbigcup\nolimits_{j\geq M,2^{j}\delta
_{N}\leq \epsilon }A_{N,j,M}(C)\right)  \notag \\
&\leq &P(2d_{\Theta }(\widehat{\theta },\theta ^{\ast })>\epsilon )+  \notag
\\
&&\dsum\nolimits_{j\geq M,2^{j}\delta _{N}\leq \epsilon }P\left(
\sup\limits_{\left( \theta ,h\right) \in A_{N,j,M}(C)}\left[ S_{N}(\theta
,h)-S_{N}(\theta ^{\ast },h)\right] \geq 0\right)  \label{b5_2}
\end{eqnarray}%
where the last inequality follows from the definition of $\widehat{\theta }$
given by (\ref{thetahat}). Since $d_{\Theta }(\widehat{\theta },\theta
^{\ast })=o_{p}(1)$, the term $P(2d_{\Theta }(\widehat{\theta },\theta
^{\ast })>\epsilon )$ tends to zero as $N\longrightarrow \infty $. Hence the
remaining part of the proof is to bound the terms in the sum (\ref{b5_2}%
).\medskip

Let $N$ be large enough such that \eqref{inequality 2} holds and choose $%
\epsilon $ to be small enough such that assumptions (\ref{quadratic
expansion}), (\ref{inequality 1}) and (\ref{maximal inequality}) hold for
every $\delta \leq \epsilon $. Note that for every sufficiently large $M$,
if $(\theta ,h)\in A_{N,j,M}(C)$, then $d_{H}^{2}(h,h^{\ast })-d_{\Theta
}^{2}(\theta ,\theta ^{\ast })\lesssim -\delta _{N}^{2}2^{2j}$ so that by (%
\ref{quadratic expansion}),%
\begin{equation}
S(\theta ,h)-S(\theta ^{\ast },h^{\ast })+e_{N}(\theta ,h)\lesssim -\delta
_{N}^{2}2^{2j}  \label{b2}
\end{equation}%
and thus
\begin{equation*}
S_{N}(\theta ,h)-S_{N}(\theta ^{\ast },h)\lesssim \left[ \widetilde{S}%
_{N}(\theta ,h)+S(\theta ^{\ast },h^{\ast })-S(\theta ^{\ast
},h)-e_{N}(\theta ,h)\right] -\delta _{N}^{2}2^{2j}.
\end{equation*}%
Therefore, by Markov inequality, each term in the sum (\ref{b5_2}) can be
bounded above by
\begin{equation}
\delta _{N}^{-2}2^{-2j}E\left[ \sup\limits_{\left( \theta ,h\right) \in
A_{N,j,M}(C)}\left\vert \widetilde{S}_{N}(\theta ,h)+S(\theta ^{\ast
},h^{\ast })-S(\theta ^{\ast },h)-e_{N}(\theta ,h)\right\vert \right] .
\label{b3}
\end{equation}%
By (\ref{inequality 2}), (\ref{inequality 1}), (\ref{maximal inequality})
and applying triangular inequality, the term (\ref{b3}) is bounded above by%
\begin{equation}
\delta _{N}^{-2}2^{-2j}\left[ N^{-1/2}\phi _{N}(2^{j}\delta
_{N})+2^{j}C\delta _{N}\varepsilon _{N}+(C\varepsilon _{N})^{2}\right] .
\label{b4}
\end{equation}%
By the monotonicity property of the mapping $\delta \mapsto \phi _{N}(\delta
)\delta ^{-\alpha }$, we have that $\phi _{N}(2^{j}\delta _{N})\leq
2^{j\alpha }\phi _{N}(\delta _{N})$. Furthermore, since $\phi _{N}(\delta
_{N})\leq \sqrt{N}\delta _{N}^{2}$, the first term in the bracket of (\ref%
{b4}) can thus be bounded by $2^{j\alpha }\delta _{N}^{2}$. Given that $%
\varepsilon _{N}\leq \delta _{N}$, the term (\ref{b4}) can be further
bounded above by $2^{j(\alpha -2)}+C2^{-j}+C^{2}2^{-2j}$. Using this fact
and the condition $\alpha <2$, it follows that the sum (\ref{b5_2}) tends to
zero as $M\longrightarrow \infty $. \medskip

Since $d_{H}(\widehat{h},h^{\ast })=O_{p}(\varepsilon _{N})$, $P(\widehat{h}%
\in H_{N}(C))$ can be made arbitrarily close to $1$ by choosing a
sufficiently large value of $C$ for every sufficiently large $N$. Therefore,
Lemma \ref{rate of convergence lemma} follows by putting together all these
results and noting that $\delta _{N}+d_{H}(\widehat{h},h^{\ast
})=O_{p}(\delta _{N})$.\medskip
\end{proof}

\section{Proof of the Rate of Convergence for $\widehat{\protect\beta }$}

\label{rate-proof-mse}

To establish the convergence rate of $\widehat{\beta }$, we apply Lemma \ref%
{rate of convergence lemma} by setting $\left( \theta ,h\right) =\left(
b,\gamma \right) $, $\left( \theta ^{\ast },h^{\ast }\right) =\left( \beta
,G\right) $, $\Theta =\{-1,1\}\times \Upsilon $, $H=\Lambda $, $s=(\tau
,d,z,x)$ and 
\begin{equation*}
m_{b,\gamma }(s)\equiv \tau (2d-1)1\{z^{\prime }b{}_{1}+\gamma (x){}^{\prime
}b_{2}>0\}.
\end{equation*}%
Assumptions (\ref{inequality 2}), (\ref{quadratic expansion}), (\ref%
{inequality 1}) and (\ref{maximal inequality}) of Lemma \ref{rate of
convergence lemma} are non-trivial and will be verified using primitive
conditions C1 - C11 of the model. Assumption (\ref{quadratic expansion}) is
concerned with the quadratic expansion of $S(b,\gamma )$ around $\left(
\beta ,G\right) $ by which we obtain the functional form of $e_{N}(b,\gamma
) $. Recall that $w=(z,G(x))$, $z=(z_{1},\widetilde{z})$, $\widetilde{w}=(%
\widetilde{z},G(x))$, $b_{1}=(b_{1,1},\widetilde{b}_{1})$, $\beta
_{1}=(\beta _{1,1},\widetilde{\beta }_{1})$, $\widetilde{b}=(\widetilde{b}%
_{1},b_{2})$ and $\widetilde{\beta }=(\widetilde{\beta }_{1},\beta _{2})$.
The following lemma will be used to establish expansion of the population
criterion $S(b,\gamma )$.\medskip

\begin{lemma}
\label{Lemma for sign of partial derivative}Under conditions C3 and C7, the
sign of $p_{1}(-\widetilde{w}^{\prime }\widetilde{\beta }{}/\beta _{1,1},%
\widetilde{z},x)$ is the same as that of $\beta _{1,1}$ for almost every $%
\left( \widetilde{z},x\right) $.
\end{lemma}

\begin{proof}
Note that the model (\ref{eq}) implies that%
\begin{equation*}
P(d=1|z,x)=F_{\varepsilon }(w^{\prime }\beta |z,x).
\end{equation*}%
Thus, by C7(ii), $P(d=1|z,x)$ is differentiable with respect to $z_{1}$ and
\begin{equation*}
\frac{\partial }{\partial z_{1}}P(d=1|z,x)=\left. \beta _{1,1}\frac{\partial
}{\partial t}F_{\varepsilon }(t|z,x)\right\vert _{t=w^{\prime }\beta
}+\left. \frac{\partial }{\partial z_{1}}F_{\varepsilon }(t|z,x)\right\vert
_{t=w^{\prime }\beta }.
\end{equation*}%
Consider the mapping $z_{1}\mapsto h(z_{1})\equiv \left. \frac{\partial }{%
\partial z_{1}}F_{\varepsilon }(t|z,x)\right\vert _{t=z_{1}\beta _{1,1}+%
\widetilde{w}^{\prime }\widetilde{\beta }}$ . By C3, $h(-\widetilde{w}%
^{\prime }\widetilde{\beta }/\beta _{1,1})=0$ for almost every $\left(
\widetilde{z},x\right) $. Therefore, Lemma \ref{Lemma for sign of partial
derivative} follows from this fact and the monotonicity of $F_{\varepsilon
}(t|z,x)$ in the argument $t$.
\end{proof}

By assumption C1, the space of the coefficient $b_{1,1}$ is $\left\{
-1,1\right\} $ and thus $b_{1,1}=\beta _{1,1}$ when $\left\Vert b-\beta
\right\Vert _{E}<\delta $ for $\delta $ small enough. Let $p(z,x)\equiv $ $%
P(d=1|z,x)$ and 
\begin{equation}
S_{1}(\widetilde{b},\gamma )\equiv E \left[ \tau (2p(z,x)-1)1\{z_{1}\beta
_{1,1}+\widetilde{z}^{\prime }\widetilde{b}{}_{1}+\gamma (x){}^{\prime
}b_{2}>0\} \right].  \label{S1}
\end{equation}%
We now derive the quadratic expansion of $S_{1}(\widetilde{b},\gamma )$
around $(\widetilde{\beta },G)$.\medskip

\begin{lemma}
\label{quadratic expansion of S(b,r)}For sufficiently small $\left\Vert 
\widetilde{b}-\widetilde{\beta }\right\Vert _{E}$ and $\left\Vert \gamma
-G\right\Vert _{\infty }$ and under conditions C3, C7, C8 and C9, we have
that 
\begin{equation*}
\left\vert S_{1}(\widetilde{\beta },\gamma )-S_{1}(\widetilde{\beta }%
,G)\right\vert \lesssim \left\Vert \gamma -G\right\Vert _{\infty }^{2}
\end{equation*}%
and there are constants $c_{1}>0$ and $c_{2}\geq 0$ such that 
\begin{equation*}
S_{1}(\widetilde{b},\gamma )-S_{1}(\widetilde{\beta },G)+e(\widetilde{b}%
,\gamma )\leq -c_{1}\left\Vert \widetilde{b}-\widetilde{\beta }\right\Vert
_{E}^{2}+c_{2}\left\Vert \gamma -G\right\Vert _{\infty }^{2}
\end{equation*}%
for some function $e(\widetilde{b},\gamma )$ that satisfies 
\begin{equation*}
\left\vert e(\widetilde{b},\gamma )\right\vert \lesssim \left\Vert 
\widetilde{b}-\widetilde{\beta }\right\Vert _{E}\left\Vert \gamma
-G\right\Vert _{\infty }\text{.}
\end{equation*}
\end{lemma}

\begin{proof}
We prove Lemma \ref{quadratic expansion of S(b,r)} explicitly for the case $%
\beta _{1,1}=1$. Proof for the case $\beta _{1,1}=-1$ can be done by similar
arguments.\medskip

Suppose now $\beta _{1,1}=1$. Then%
\begin{eqnarray*}
&&S_{1}(\widetilde{b},\gamma )-S_{1}\left( \widetilde{\beta },G\right) \\
&=&E\left( \tau (2p(z,x)-1)\left[ 1\{z_{1}+\widetilde{z}^{\prime }\widetilde{%
\beta }_{1}+G(x){}^{\prime }\beta _{2}\leq 0\}-1\{z_{1}+\widetilde{z}%
^{\prime }\widetilde{b}{}_{1}+\gamma (x){}^{\prime }b_{2}\leq 0\}\right]
\right) .
\end{eqnarray*}%
Let
\begin{eqnarray*}
\lambda (t) &\equiv &\widetilde{z}^{\prime }\left( \widetilde{\beta }{}%
_{1}+t\left( \widetilde{b}_{1}-\widetilde{\beta }{}_{1}\right) \right)
+\left( G(x)+t\left( \gamma (x)-G(x)\right) {}\right) ^{\prime }\left( \beta
_{2}+t\left( b_{2}-\beta _{2}\right) \right) , \\
\Psi (t) &\equiv &-E(\tau (2p(z,x)-1)1\{z_{1}+\lambda (t)\leq 0\}).
\end{eqnarray*}%
The first-order and second-order derivatives of $\Psi (t)$ are derived as
follows:%
\begin{eqnarray*}
\Psi ^{\prime }(t) &=&E\left( \tau \lambda ^{\prime }(t)\left( 2p(-\lambda
(t),\widetilde{z},x)-1\right) g_{1}(-\lambda (t)|\widetilde{z},x)\right) , \\
\Psi ^{\prime \prime }(t) &=&-E\left\{ \tau \left( \lambda ^{\prime
}(t)\right) ^{2}\left[ 2p_{1}\left( -\lambda (t),\widetilde{z},x\right)
g_{1}(-\lambda (t)|\widetilde{z},x)\right. \right. \\
&&+\left. \left. \left( 2p(-\lambda (t),\widetilde{z},x)-1\right) \frac{%
\partial }{\partial z_{1}}g_{1}(-\lambda \left( t\right) |\widetilde{z},x)%
\right] \right\} \\
&&+E\left( 2\tau \left[ \left( 2p(-\lambda (t),\widetilde{z},x)-1\right) %
\right] g_{1}(-\lambda (t)|\widetilde{z},x)(\gamma (x)-G(x))^{\prime
}(b_{2}-\beta _{2})\right) .
\end{eqnarray*}%
Then the second order expansion of $S_{1}(\widetilde{b},\gamma )-S_{1}\left(
\widetilde{\beta },G\right) $ takes the form
\begin{equation*}
\Psi ^{\prime }(0)+\Psi ^{\prime \prime }(0)/2+o\left( \left( \max \left\{
\left\Vert \widetilde{b}-\widetilde{\beta }\right\Vert _{E},\left\Vert
\gamma -G\right\Vert _{\infty }\right\} \right) ^{2}\right)
\end{equation*}%
where by C7 and C8, the remainder term has the stated order uniformly over $%
\widetilde{b}$ and $\gamma $. Given assumption C3, it follows that $p(-%
\widetilde{w}^{\prime }\widetilde{\beta },\widetilde{z},x)=1/2$ for almost
every $\left( \widetilde{z},x\right) $. Let
\begin{equation*}
\kappa (\widetilde{z},x)=2p_{1}\left( -\widetilde{w}^{\prime }\widetilde{%
\beta },\widetilde{z},x\right) g_{1}(-\widetilde{w}^{\prime }\widetilde{%
\beta }|\widetilde{z},x).
\end{equation*}%
Then we have that
\begin{eqnarray*}
\Psi ^{\prime }(0)+\Psi ^{\prime \prime }(0)/2 &=&-E\left( \tau \kappa (%
\widetilde{z},x)\left( \widetilde{w}^{\prime }(\widetilde{b}-\widetilde{%
\beta }{})+(\gamma (x)-G(x))^{\prime }\beta _{2}\right) ^{2}\right) \\
&=&-\left( A_{1}+A_{2}+e(\widetilde{b},\gamma )\right) ,
\end{eqnarray*}%
where
\begin{eqnarray}
A_{1}(\widetilde{b}) &\equiv &(\widetilde{b}-\widetilde{\beta })^{\prime
}E(\tau \kappa (\widetilde{z},x)\widetilde{w}\widetilde{w}^{\prime })(%
\widetilde{b}-\widetilde{\beta }),  \label{A1} \\
A_{2}(\gamma ) &\equiv &E\left( \tau \kappa (\widetilde{z},x)\left( \gamma
(x)-G(x)\right) ^{\prime }\beta _{2}\beta _{2}^{\prime }\left( \gamma
(x)-G(x)\right) \right) ,  \label{A2} \\
e(\widetilde{b},\gamma ) &\equiv &2(\widetilde{b}-\widetilde{\beta }%
)^{\prime }E\left( \tau \kappa (\widetilde{z},x)\widetilde{w}\beta
_{2}^{\prime }\left( \gamma (x)-G(x)\right) \right) .  \label{e(b,r)}
\end{eqnarray}%
Under condition C9, $E(\tau \kappa (\widetilde{z},x)\widetilde{w}\widetilde{w%
}^{\prime })$ is positive definite, so that $A_{1}\geq c_{1}\left\Vert
\widetilde{b}-\widetilde{\beta }\right\Vert _{E}^{2}$ for some positive real
constant $c_{1}$. By Lemma \ref{Lemma for sign of partial derivative}, $%
p_{1}\left( -\widetilde{w}^{\prime }\widetilde{\beta },\widetilde{z}%
,x\right) \geq 0$ and thus $\kappa (\widetilde{z},x)\geq 0$. By
Cauchy-Schwarz inequality, $0\leq A_{2}\leq c_{2}\left\Vert \gamma
-G\right\Vert _{\infty }^{2}$, where $c_{2}\equiv E(\tau \kappa (\widetilde{z%
},x))\left\Vert \beta _{2}\right\Vert _{E}^{2}\geq 0$, and the function $e(%
\widetilde{b},\gamma )$ satisfies that%
\begin{eqnarray*}
\left\vert e(\widetilde{b},\gamma )\right\vert &\leq &2E\left( \tau \kappa (%
\widetilde{z},x)\left\vert (\widetilde{b}-\widetilde{\beta })^{\prime }%
\widetilde{w}\beta _{2}^{\prime }\left( \gamma (x)-G(x)\right) \right\vert
\right) \\
&\leq &2E\left( \tau \kappa (\widetilde{z},x)\left\Vert \widetilde{w}%
\right\Vert _{E}\right) \left\Vert \beta _{2}\right\Vert _{E}\left\Vert
\widetilde{b}-\widetilde{\beta }\right\Vert _{E}\left\Vert \gamma
-G\right\Vert _{\infty }.
\end{eqnarray*}%
Hence Lemma \ref{quadratic expansion of S(b,r)} follows by noting that when $%
\left\Vert \widetilde{b}-\widetilde{\beta }\right\Vert _{E}$ and $\left\Vert
\gamma -G\right\Vert _{\infty }$ are sufficiently small,
\begin{equation*}
\left\vert S_{1}(\widetilde{\beta },\gamma )-S_{1}(\widetilde{\beta }%
,G)\right\vert =\left\vert A_{2}+o\left( \left\Vert \gamma -G\right\Vert
_{\infty }^{2}\right) \right\vert \leq c_{2}\left\Vert \gamma -G\right\Vert
_{\infty }^{2}
\end{equation*}%
and
\begin{eqnarray*}
S_{1}(\widetilde{b},\gamma )-S_{1}\left( \widetilde{\beta },G\right) +e(%
\widetilde{b},\gamma ) &\leq &-A_{1}+A_{2} \\
&\leq &-c_{1}\left\Vert \widetilde{b}-\widetilde{\beta }\right\Vert
_{E}^{2}+c_{2}\left\Vert \gamma -G\right\Vert _{\infty }^{2}\text{.}
\end{eqnarray*}
\end{proof}

We now verify assumption (\ref{maximal inequality}) of Lemma \ref{rate of
convergence lemma}. Note that for $\delta $ sufficiently small, assumption
C1 implies that $b_{1,1}=\beta _{1,1}$ when $\left\Vert b-\beta \right\Vert
_{E}\leq \delta $. Therefore we can focus on analyzing (\ref{maximal
inequality}) for the case of $b_{1,1}=\beta _{1,1}$ and $\left\Vert 
\widetilde{b}-\widetilde{\beta }\right\Vert _{E}\leq \delta $. For any $%
s=\left( \tau ,d,z,x\right) $, consider the following recentered function 
\begin{equation}
\widetilde{m}_{\widetilde{b},\gamma }(s)\equiv \tau (2d-1)\left[
1\{z_{1}\beta _{1,1}+\widetilde{z}^{\prime }\widetilde{b}{}_{1}+\gamma
(x){}^{\prime }b_{2}>0\}-1\{z_{1}\beta _{1,1}+\widetilde{z}^{\prime }%
\widetilde{\beta }_{1}+\gamma (x){}^{\prime }\beta _{2}>0\}\right]
\label{recentered m(b,r)}
\end{equation}%
and the class of functions%
\begin{equation}
\digamma _{\delta ,\varepsilon }\equiv \left\{ \widetilde{m}_{\widetilde{b}%
,\gamma }:\left\Vert \widetilde{b}-\widetilde{\beta }\right\Vert _{E}\leq
\delta ,\left\Vert \gamma -G\right\Vert _{\alpha ,p}\leq \varepsilon
\right\} .  \label{F(delta.epsilon)}
\end{equation}%
Let $\left\Vert \cdot \right\Vert _{L_{r}(P)}$ denote the $L_{r}(P)$ norm
such that $\left\Vert f\right\Vert _{L_{r}(P)}\equiv \left[ E(\left\vert
f(\tau ,d,z,x)\right\vert ^{r})\right] ^{1/r}$ for any measurable function $%
f $. For any $\epsilon >0$, let $N_{[]}(\epsilon ,\digamma ,L_{r}(P))$
denote the $L_{r}(P)$ - bracketing number for a given function space $%
\digamma $. Namely, $N_{[]}(\epsilon ,\digamma ,L_{r}(P))$ is the minimum
number of $L_{r}(P)$ - brackets of length $\epsilon $ required to cover $%
\digamma $ (see e.g., van der Vaart (1998, p.\thinspace 270)). The logarithm
of bracketing number for $\digamma $ is referred to as the bracketing
entropy for $\digamma $. Assumption (\ref{maximal inequality}) is a
stochastic equicontinuity condition concerning the complexity of the
function space $\digamma _{\delta ,\varepsilon }$ in terms of its envelope
function and bracketing entropy. Let $M_{_{\delta ,\varepsilon }}$ denote an
envelope for $\digamma _{\delta ,\varepsilon }$ such that $\left\vert 
\widetilde{m}_{\widetilde{b},\gamma }(s)\right\vert $ $\leq \left\vert
M_{_{\delta ,\varepsilon }}(s)\right\vert $ for all $s$ and for all $%
\widetilde{m}_{\widetilde{b},\gamma }\in \digamma _{\delta ,\varepsilon }$.
The next lemma derives the envelope function $M_{_{\delta ,\varepsilon }}$%
.\medskip

\begin{lemma}
\label{envelope}Let $\delta $ and $\varepsilon $ be sufficiently small. Then
under conditions C1, C4 ,C6 and C10, for some real constants $a_{1}>0$ and $%
a_{2}>0$, we can take 
\begin{equation*}
M_{_{\delta ,\varepsilon }}=1\{a_{1}\max \{\delta ,\varepsilon \}\geq
\left\vert w^{\prime }\beta \right\vert \}
\end{equation*}%
and furthermore, 
\begin{equation}
\left\Vert M_{_{\delta ,\varepsilon }}\right\Vert _{L_{2}(P)}\leq a_{2}\sqrt{%
\max \{\delta ,\varepsilon \}}.  \label{bound}
\end{equation}
\end{lemma}

\begin{proof}
Note that%
\begin{align*}
\lefteqn{\left\vert \widetilde{m}_{\widetilde{b},\gamma }(\tau
,d,z,x)\right\vert } \\
\leq & \text{ }1\{z_{1}\beta _{1,1}+\widetilde{z}^{\prime }\widetilde{b}{}%
_{1}+\gamma (x){}^{\prime }b_{2}>0\geq z^{\prime }\beta {}_{1}+\gamma
(x){}^{\prime }\beta _{2}\;\;\text{ or } \\
\;\;\;\;\;\;\;\;\;\;\;\;\;\;\;\;\;\;\;\;& \;z^{\prime }\beta {}_{1}+\gamma
(x){}^{\prime }\beta _{2}>0\geq z_{1}\beta _{1,1}+\widetilde{z}^{\prime }%
\widetilde{b}{}_{1}+\gamma (x){}^{\prime }b_{2}\}.
\end{align*}%
Under condition C6, there is a positive real constant $B$ such that $%
\left\Vert \widetilde{z}\right\Vert _{E}<B$ with probability 1.\ Hence if $%
\left\Vert \widetilde{b}-\beta \right\Vert _{E}\leq \delta $ and $\left\Vert
\gamma -G\right\Vert _{\alpha ,p}\leq \varepsilon $, then we have that%
\begin{eqnarray*}
&&z_{1}\beta _{1,1}+\widetilde{z}^{\prime }\widetilde{b}{}_{1}+\gamma
(x){}^{\prime }b_{2}>0\geq z^{\prime }\beta {}_{1}+\gamma (x){}^{\prime
}\beta _{2} \\
&\iff &\widetilde{z}^{\prime }(\widetilde{b}{}_{1}-\widetilde{\beta }{}%
_{1})+\gamma (x){}^{\prime }(b_{2}-\beta _{2})>-\left[ z^{\prime }\beta
{}_{1}+\gamma (x){}^{\prime }\beta _{2}\right] \geq 0 \\
&\Longrightarrow &\delta \left[ \left\Vert \widetilde{z}\right\Vert
_{E}+\left\Vert \gamma \right\Vert _{\infty }\right] \geq -\left[ z^{\prime
}\beta {}_{1}+\gamma (x){}^{\prime }\beta _{2}\right] \text{ and }0\geq
w^{\prime }\beta {}+(\gamma (x)-G(x)){}^{\prime }\beta _{2} \\
&\Longrightarrow &w^{\prime }\beta +(\gamma (x)-G(x{}))^{\prime }\beta
_{2}\geq -\delta \left[ \left\Vert \widetilde{z}\right\Vert _{E}+\varepsilon
+\left\Vert G\right\Vert _{\infty }\right] \text{ and }\varepsilon
\left\Vert \beta _{2}\right\Vert _{E}\geq w^{\prime }\beta {} \\
&\Longrightarrow &\delta \left[ B+\varepsilon +\left\Vert G\right\Vert
_{\infty }\right] +\varepsilon \left\Vert \beta _{2}\right\Vert _{E}\geq
w^{\prime }\beta \geq -\delta \left[ B+\varepsilon +\left\Vert G\right\Vert
_{\infty }\right] -\varepsilon \left\Vert \beta _{2}\right\Vert _{E}
\end{eqnarray*}%
Based on similar arguments, it also follows that
\begin{eqnarray*}
&&z^{\prime }\beta {}_{1}+\gamma (x){}^{\prime }\beta _{2}>0\geq z_{1}\beta
_{1,1}+\widetilde{z}^{\prime }\widetilde{b}{}_{1}+\gamma (x){}^{\prime }b_{2}
\\
&\Longrightarrow &\delta \left[ B+\varepsilon +\left\Vert G\right\Vert
_{\infty }\right] +\varepsilon \left\Vert \beta _{2}\right\Vert _{E}\geq
w^{\prime }\beta \geq -\delta \left[ B+\varepsilon +\left\Vert G\right\Vert
_{\infty }\right] -\varepsilon \left\Vert \beta _{2}\right\Vert _{E}
\end{eqnarray*}%
Therefore, Lemma \ref{envelope} follows by noting that for $\varepsilon $
sufficiently small (e.g., $\varepsilon <1$), we can take
\begin{equation*}
M_{_{\delta ,\varepsilon }}=1\{a_{1}\max \{\delta ,\varepsilon \}\geq
\left\vert w^{\prime }\beta \right\vert \}
\end{equation*}%
where $a_{1}\equiv 2\max \{\left( B+1+\left\Vert G\right\Vert _{\infty
}\right) ,\left\Vert \beta _{2}\right\Vert _{E}\}$. By C1 and C10, $%
0<a_{1}<\infty $ and hence by C4, $\left\Vert M_{_{\delta ,\varepsilon
}}\right\Vert _{L_{2}(P)}\leq a_{2}\sqrt{\max \{\delta ,\varepsilon \}}$
with $a_{2}\equiv \sqrt{2a_{1}L}$ where $L$ is the positive constant stated
in condition C4.\medskip
\end{proof}

The following lemma establishes the bound for the bracketing entropy for $%
\digamma _{\delta ,\varepsilon }$. \medskip

\begin{lemma}
\label{entropy bound}Given conditions C1, C4, C6, C7, C8 and C10, we have
that for sufficiently small $\delta $ and $\varepsilon $, 
\begin{equation*}
\log N_{[]}(\epsilon ,\digamma _{\delta ,\varepsilon },L_{2}(P))\lesssim
(\max \{\delta ,\varepsilon \})^{q/\alpha }\epsilon ^{-2q/\alpha }.
\end{equation*}
\end{lemma}

\begin{proof}
For $j\in \{1,...,p\}$, let $\widetilde{\Lambda }_{j}(\varepsilon )$ and $%
\widetilde{\Lambda }_{j}\mathbf{B}_{j}(\delta ,\varepsilon )$ be classes of
functions defined as%
\begin{eqnarray*}
\widetilde{\Lambda }_{j}(\varepsilon ) &\equiv &\left\{ (\gamma
_{j}-G_{j})/\varepsilon :\left\Vert \gamma _{j}-G_{j}\right\Vert _{\alpha
}\leq \varepsilon \right\} , \\
\widetilde{\Lambda }_{j}\mathbf{B}_{j}(\delta ,\varepsilon ) &\equiv
&\left\{ (\gamma _{j}(x){}-G_{j}(x))(b_{2,j}-\beta _{2,j})/\left(
\varepsilon \delta \right) :\left\Vert \gamma _{j}-G_{j}\right\Vert _{\alpha
}\leq \varepsilon ,\left\vert b_{2,j}-\beta _{2,j}\right\vert \leq \delta
\right\} .
\end{eqnarray*}%
Assumption C10 implies that both $\widetilde{\Lambda }_{j}(\varepsilon )$
and $\widetilde{\Lambda }_{j}\mathbf{B}_{j}(\delta ,\varepsilon )$ are $%
C_{1}^{\alpha }$. By Corollary 2.7.2 of van der Vaart and Wellner (1996,
p.\thinspace 157), we have that for $j\in \{1,...,p\}$,
\begin{equation}
\log N_{[]}(\epsilon ^{2},\widetilde{\Lambda }_{j}(\varepsilon
),L_{1}(P))\lesssim \epsilon ^{-2q/\alpha }\text{ and }\log N_{[]}(\epsilon
^{2},\widetilde{\Lambda }_{j}\mathbf{B}_{j}(\delta ,\varepsilon
),L_{1}(P))\lesssim \epsilon ^{-2q/\alpha }.
\label{entropy bound C(alpha,1)}
\end{equation}

Note that for $s=(\tau ,d,z,x)$, $\widetilde{m}_{\widetilde{b},\gamma }(s)$
defined by (\ref{recentered m(b,r)}) can be rewritten as
\begin{equation*}
\widetilde{m}_{\widetilde{b},\gamma }(s)=\tau d\left[ 1\{h(s;\widetilde{b}%
)>0\}-1\{h(s;\widetilde{\beta })>0\}\right] +\tau (1-d)\left[ 1\{h(s;%
\widetilde{b})\leq 0\}-1\{h(s;\widetilde{\beta })\leq 0\}\right]
\end{equation*}%
where
\begin{eqnarray*}
h(s;\widetilde{b}) &\equiv &w^{\prime }\beta +\widetilde{w}^{\prime }(%
\widetilde{b}-\widetilde{\beta })+(\gamma (x){}-G(x))^{\prime }(b_{2}-\beta
_{2})+(\gamma (x){}-G(x))^{\prime }\beta _{2}, \\
h(s;\widetilde{\beta }) &\equiv &w^{\prime }\beta +(\gamma
(x){}-G(x))^{\prime }\beta _{2}.
\end{eqnarray*}%
Consider the following spaces:%
\begin{eqnarray*}
\Theta _{1} &\equiv &\{\widetilde{w}^{\prime }(\widetilde{b}-\widetilde{%
\beta }):\left\Vert \widetilde{b}-\widetilde{\beta }\right\Vert _{E}\leq
\delta \}, \\
\Theta _{2,j} &\equiv &\left\{ (\gamma _{j}(x){}-G_{j}(x))(b_{2,j}-\beta
_{2,j}):\left\Vert \gamma _{j}-G_{j}\right\Vert _{\alpha }\leq \varepsilon
,\left\vert b_{2,j}-\beta _{2,j}\right\vert \leq \delta \right\} \text{ for }%
j\in \{1,...,p\}, \\
\Theta _{2} &\equiv &\left\{ (\gamma (x){}-G(x))^{\prime }(b_{2}-\beta
_{2}):\left\Vert \gamma -G\right\Vert _{\alpha ,p}\leq \varepsilon \text{,}%
\left\Vert b_{2}-\beta _{2}\right\Vert _{E}\leq \delta \right\} , \\
\Theta _{3,j} &\equiv &\{(\gamma _{j}(x){}-G_{j}(x))\beta _{2,j}:\left\Vert
\gamma _{j}-G_{j}\right\Vert _{\alpha }\leq \varepsilon \}\text{ for }j\in
\{1,...,p\}, \\
\Theta _{3} &\equiv &\{(\gamma (x){}-G(x))^{\prime }\beta _{2}:\left\Vert
\gamma -G\right\Vert _{\alpha ,p}\leq \varepsilon \}, \\
\Theta _{4} &\equiv &\{h(\tau ,d,z,x;\widetilde{b})-w^{\prime }\beta
:\left\Vert \gamma -G\right\Vert _{\alpha ,p}\leq \varepsilon \text{,}%
\left\Vert \widetilde{b}-\widetilde{\beta }\right\Vert _{E}\leq \delta \}.
\end{eqnarray*}%
Let $n_{i}(\epsilon )\equiv \log N_{[]}(\epsilon ,\Theta _{i},L_{1}(P))$ for
$i\in \{1,2,3,4\}$ and $n_{k,j}(\epsilon )\equiv \log N_{[]}(\epsilon
,\Theta _{k,j},L_{1}(P))$ for $\left( k,j\right) \in \{2,3\}\times
\{1,...,p\}$. Let $\psi \equiv \sqrt{\max \{\delta ,\varepsilon \}}$.$%
\medskip $

Since $\Theta _{1}$ is a pointwise Lipschitz class of functions with
envelope $\left\Vert \widetilde{w}\right\Vert _{E}\delta $. By condition C8,
$E(\left\Vert \widetilde{w}\right\Vert _{E})$ is finite. Thus applying
Theorem 2.7.11 of van der Vaart and Wellner (1996, p.\thinspace 164), we
have that%
\begin{equation}
n_{1}(\epsilon ^{2})\lesssim \frac{q}{\alpha }\log (\delta /\epsilon
^{2})\lesssim \delta ^{q/\alpha }\epsilon ^{-2q/\alpha }\lesssim \psi
^{2q/\alpha }\epsilon ^{-2q/\alpha }\text{.}  \label{n1}
\end{equation}%
Note that for any norm $\left\Vert \cdot \right\Vert $, any fixed real
valued $c$, any class of functions $\digamma $, it is straightforward to
verify that
\begin{eqnarray}
&&N_{[]}(\epsilon ,c\digamma ,\left\Vert \cdot \right\Vert )=1\text{ for }c=0
\notag \\
&&N_{[]}(\epsilon ,c\digamma ,\left\Vert \cdot \right\Vert )\leq
N_{[]}(\epsilon /\left\vert c\right\vert ,\digamma ,\left\Vert \cdot
\right\Vert )\text{ for }c\neq 0  \notag
\end{eqnarray}%
where $c\digamma \equiv \{cf:f\in \digamma \}$. \medskip

Using this fact, we have that $n_{2,j}(\epsilon ^{2})=\log N_{[]}(\epsilon
^{2}/(\varepsilon \delta ),\widetilde{\Lambda }_{j}\mathbf{B}_{j}(\delta
,\varepsilon ),L_{1}(P))$ and $n_{3,j}(\epsilon ^{2})=0$ for $\beta _{2,j}=0$
and $n_{3,j}(\epsilon ^{2})\leq \log N_{[]}(\epsilon ^{2}/(\varepsilon
\left\vert \beta _{2,j}\right\vert ),\widetilde{\Lambda }_{j}(\varepsilon
),L_{1}(P))$ for $\beta _{2,j}\neq 0$. Hence for sufficiently small $\delta $
and $\varepsilon $ (e.g., $\delta <1$ and $\varepsilon <1$) and by (\ref%
{entropy bound C(alpha,1)}), it follows that
\begin{equation*}
n_{2,j}(\epsilon ^{2})\leq \log N_{[]}(\epsilon ^{2}\psi ^{-2},\widetilde{%
\Lambda }_{j}\mathbf{B}_{j}(\delta ,\varepsilon ),L_{1}(P))\lesssim \psi
^{2q/\alpha }\epsilon ^{-2q/\alpha }.
\end{equation*}%
Using similar arguments, we can also deduce that $n_{3,j}(\epsilon
^{2})\lesssim \psi ^{2q/\alpha }\epsilon ^{-2q/\alpha }$. \medskip

By preservation of bracketing metric entropy (see, e.g., Lemma 9.25 of
Kosorok (2008, p.\thinspace 169)), we have that for $i\in \{2,3\}$,
\begin{equation*}
n_{i}(\epsilon )\leq n_{i,p}(\epsilon
2^{1-p})+\sum\nolimits_{j=1}^{p-1}n_{i,j}(\epsilon 2^{-j}).
\end{equation*}%
and $n_{4}(\epsilon )\leq n_{1}(\epsilon /2)+n_{2}(\epsilon
/4)+n_{3}(\epsilon /4)$. Therefore by the bounds derived above, it follows
that $n_{2}(\epsilon ^{2})\lesssim \psi ^{2q/\alpha }\epsilon ^{-2q/\alpha }$%
, $n_{3}(\epsilon ^{2})\lesssim \psi ^{2q/\alpha }\epsilon ^{-2q/\alpha }$
and also $n_{4}(\epsilon ^{2})\lesssim \psi ^{2q/\alpha }\epsilon
^{-2q/\alpha }$.\medskip

Now let $f_{1}^{L}\leq f_{1}^{U},...,f_{N_{[]}(\epsilon ^{2},\Theta
_{3},L_{1}(P))}^{L}\leq f_{N_{[]}(\epsilon ^{2},\Theta _{3},L_{1}(P))}^{U}$
and $g_{1}^{L}\leq g_{1}^{U},...,g_{N_{[]}(\epsilon ^{2},\Theta
_{4},L_{1}(P))}^{L}\leq g_{N_{[]}(\epsilon ^{2},\Theta _{4},L_{1}(P))}^{U}$
be the $\epsilon ^{2}$-brackets with bracket length defined by $L_{1}(P)$
for the spaces $\Theta _{3}$ and $\Theta _{4}$, respectively. For $1\leq
k\leq N_{[]}(\epsilon ^{2},\Theta _{3},L_{1}(P))$ and $1\leq j\leq
N_{[]}(\epsilon ^{2},\Theta _{4},L_{1}(P))$, define%
\begin{eqnarray*}
m_{jk}^{L}(\tau ,d,z,x) &\equiv &\tau d\left[ 1\{w^{\prime }\beta
+g_{j}^{L}(z,x)>0\}-1\{w^{\prime }\beta +f_{k}^{U}(z,x){}>0\}\right] \\
&&+\tau (1-d)\left[ 1\{w^{\prime }\beta +g_{j}^{U}(z,x)\leq 0\}-1\{w^{\prime
}\beta +f_{k}^{L}(z,x){}\leq 0\}\right] , \\
m_{jk}^{U}(\tau ,d,z,x) &\equiv &\tau d\left[ 1\{w^{\prime }\beta
+g_{j}^{U}(z,x)>0\}-1\{w^{\prime }\beta +f_{k}^{L}(z,x){}>0\}\right] \\
&&+\tau (1-d)\left[ 1\{w^{\prime }\beta +g_{j}^{L}(z,x)\leq 0\}-1\{w^{\prime
}\beta +f_{k}^{U}(z,x){}\leq 0\}\right] .
\end{eqnarray*}%
Note that%
\begin{equation*}
0\leq m_{jk}^{U}-m_{jk}^{L}\leq 2\left[ 1\{g_{j}^{L}\leq -w^{\prime }\beta
<g_{j}^{U}\}+1\{f_{k}^{L}\leq -w^{\prime }\beta <f_{k}^{U}\}\right] \text{.}
\end{equation*}%
Thus%
\begin{equation}
E\left( m_{jk}^{U}-m_{jk}^{L}\right) ^{2}\leq 12P(g_{j}^{L}\leq -w^{\prime
}\beta <g_{j}^{U})+4P(f_{k}^{L}\leq -w^{\prime }\beta <f_{k}^{U}).
\label{m_bound}
\end{equation}%
By condition C1 and given $(\widetilde{z},x)$, the mapping $z_{1}\longmapsto
w^{\prime }\beta $ is one-to-one. Hence by condition C7, the density of $%
w^{\prime }\beta $ conditional on $(\widetilde{z},x)$ is bounded and by (\ref%
{m_bound}), it then follows that $\left\Vert
m_{jk}^{U}-m_{jk}^{L}\right\Vert _{L_{2}(P)}\lesssim \epsilon $. Moreover
for each $\widetilde{m}_{\widetilde{b},\gamma }\in \digamma _{\delta
,\varepsilon _{N}}$, there is a bracket $\left[ m_{jk}^{L},m_{jk}^{U}\right]
$ in which it lies. Therefore,
\begin{equation*}
\log N_{[]}(\epsilon ,\digamma _{\delta ,\varepsilon _{N}},L_{2}(P))\lesssim
n_{3}(\epsilon ^{2})+n_{4}(\epsilon ^{2})\lesssim \psi ^{2q/\alpha }\epsilon
^{-2q/\alpha }\text{.}
\end{equation*}
\end{proof}

Replacing $\left( \theta ,h\right) $ and $\theta ^{\ast }$ with $((\beta
_{1,1},\widetilde{b}),\gamma )$ and $(\beta _{1,1},\widetilde{\beta })$,
respectively in the definition of $\widetilde{S}_{N}$ given by (\ref%
{recentered criterion}), we now verify assumption (\ref{maximal inequality})
in the next lemma.\medskip

\begin{lemma}
\label{verify maximal inequality}For sufficiently small $\delta $ and $%
\varepsilon $, under conditions C1, C4, C6, C7, C8 and C10, 
\begin{equation*}
E\left[ \sup_{\left\Vert \widetilde{b}-\widetilde{\beta }\right\Vert
_{E}\leq \delta ,\left\Vert \gamma -G\right\Vert _{\alpha ,p}\leq
\varepsilon }\left\vert \widetilde{S}_{N}(\widetilde{b},\gamma )\right\vert %
\right] \lesssim \frac{\sqrt{\max \{\delta ,\varepsilon \}}}{\sqrt{N}}.
\end{equation*}
\end{lemma}

\begin{proof}
Let $\psi \equiv \sqrt{\max \{\delta ,\varepsilon \}}$. By Lemmas \ref%
{envelope} and \ref{entropy bound}, we have that
\begin{equation*}
\int\nolimits_{0}^{\left\Vert M_{_{\delta ,\varepsilon }}\right\Vert
_{L_{2}(P)}}\sqrt{\log N_{[]}(\epsilon ,\digamma _{\delta ,\varepsilon
},L_{2}(P))}d\epsilon \lesssim \psi ^{q/\alpha }\int\nolimits_{0}^{a_{2}\psi
}\epsilon ^{-q/\alpha }d\epsilon \lesssim \psi
\end{equation*}%
where the last inequality follows since $\alpha >q$. Lemma \ref{verify
maximal inequality} hence follows by applying Corollary 19.35 of van der
Vaart (1998, p.\thinspace 288).\medskip
\end{proof}

We now prove Theorem \ref{thm-rate}. \medskip

\begin{proof}[Proof of Theorem \protect\ref{thm-rate}]
We take $\delta _{N}=N^{-1/3}$, $d_{\Theta }(b,\beta )=\sqrt{c_{1}}%
\left\Vert b-\beta \right\Vert _{E}$ and $d_{H}(\gamma ,G)=\sqrt{c_{2}}%
\left\Vert \gamma -G\right\Vert _{\alpha ,p}$ in the application of Lemma %
\ref{rate of convergence lemma}, where $c_{1}$ and $c_{2}$ are real
constants stated in Lemma \ref{quadratic expansion of S(b,r)}. \medskip

Since $c_{1}>0$, the norm by the metric $d_{\Theta }(\cdot ,\cdot )$ is
equivalent to the Euclidean norm and thus by Theorem \ref{thm-consist}, $%
d_{\Theta }(\widehat{\beta },\beta )=o_{p}(1)$. Moreover since $c_{2}\geq 0$%
, assumption C11 implies that $d_{H}(\widehat{G},G)=O_{p}(\varepsilon _{N})$%
. Given assumption C1, for sufficiently small $\delta $, we have that $%
b_{1,1}=\beta _{1,1}$ when $d_{\Theta }(b,\beta )\leq \delta $. Hence for
sufficiently small $\delta $ and $\varepsilon _{N}$, by Lemma \ref{quadratic
expansion of S(b,r)} and noting that $\left\Vert \cdot \right\Vert _{\alpha
,p}$ is stronger than $\left\Vert \cdot \right\Vert _{\infty }$, assumptions
(\ref{inequality 2}), (\ref{quadratic expansion}) and (\ref{inequality 1})
hold. \medskip

By Lemma \ref{verify maximal inequality} and by taking $C$ sufficiently
large in the definition of $H_{N}(C)$ of Lemma \ref{rate of convergence
lemma}, assumption (\ref{maximal inequality}) also holds with $\phi
_{N}(\delta )=\sqrt{\max \{\delta ,\varepsilon _{N}\}}$. Clearly, $\phi
_{N}(\delta )\delta ^{-\alpha }$ is decreasing for some $\alpha <2$. By
assumption C11, $\varepsilon _{N}\leq \delta _{N}$ and thus $\phi
_{N}(\delta _{N})\leq \sqrt{N}\delta _{N}^{2}$ for every $N$. Therefore, all
conditions stated in Lemma \ref{rate of convergence lemma} are fulfilled and
the result of Theorem \ref{thm-rate} hence follows.
\end{proof}

\section{Proof of Asymptotic Distribution of $N^{1/3}(\widehat{\protect\beta 
}-\protect\beta )$}

For $C>0$, define the sets 
\begin{eqnarray*}
\Theta _{N}(C) &\equiv &\left\{ b\in \Theta :N^{1/3}\left\Vert b-\beta
\right\Vert _{E}\leq C\right\} , \\
H_{N}(C) &\equiv &\{\gamma \in \Lambda :\left\Vert \gamma -G\right\Vert
_{\alpha ,p}\leq C\varepsilon _{N}\}
\end{eqnarray*}%
where $\varepsilon _{N}$ is the sequence stated in the assumptions of
Theorem \ref{thm-distribution}. For each $\left( b,\gamma \right) $, define
the following recentered empirical and population criterion functions 
\begin{eqnarray*}
\overline{S}_{N}\left( b,\gamma \right) &\equiv &S_{N}\left( b,\gamma
\right) -S_{N}\left( \beta ,\gamma \right) , \\
\overline{S}\left( b,\gamma \right) &\equiv &S\left( b,\gamma \right)
-S\left( \beta ,\gamma \right) .
\end{eqnarray*}%
Clearly, $\widehat{\beta }$ and $\widehat{\beta }$ $_{G}$ , defined by (\ref%
{beta_hat}) and (\ref{beta_hat_G}), are still maximizers of the objective
functions $\overline{S}_{N}(b,\widehat{G})$ and $\overline{S}_{N}(b,G)$,
respectively. Decompose $\overline{S}_{N}(b,\widehat{G})-\overline{S}%
_{N}(b,G)$ as follows.%
\begin{equation}
\overline{S}_{N}(b,\widehat{G})-\overline{S}_{N}(b,G)=\left[ \widetilde{S}%
_{N}(b,\widehat{G})-\widetilde{S}_{N}(b,G)\right] +\left[ \overline{S}(b,%
\widehat{G})-\overline{S}(b,G)\right]  \label{decomposition}
\end{equation}%
where%
\begin{equation*}
\widetilde{S}_{N}(b,\gamma )\equiv \overline{S}_{N}\left( b,\gamma \right) -%
\overline{S}\left( b,\gamma \right) .
\end{equation*}%
We shall need the following results.\medskip

For $\delta >0$ and $\varepsilon >0$, Consider the local neighborhoods $%
\Theta (\delta )$ and $H(\varepsilon )$ defined as%
\begin{eqnarray}
\Theta (\delta ) &\equiv &\{b\in \Theta :\left\Vert b-\beta \right\Vert
_{E}\leq \delta \},  \label{theta_delta} \\
H(\varepsilon ) &\equiv &\{\gamma \in \Lambda :\left\Vert \gamma
-G\right\Vert _{\alpha ,p}\leq \varepsilon \}.  \label{H_epsilon}
\end{eqnarray}%
Recall that $w\equiv (z,G(x))$, $z\equiv (z_{1},\widetilde{z})$, $\widetilde{%
w}\equiv (\widetilde{z},G(x))$, $\widetilde{b}\equiv (\widetilde{b}%
_{1},b_{2})$ and $\widetilde{\beta }\equiv (\widetilde{\beta }_{1},\beta
_{2})$. Note that for $\delta $ sufficiently small, assumption C1 implies
that $b_{1,1}=\beta _{1,1}$ when $b\in \Theta (\delta )$. Therefore we may
assume that $\Theta (\delta )=\{b\in \Theta :$ $b_{1,1}=\beta _{1,1}$ and $%
\widetilde{b}\in \widetilde{\Theta }(\delta )\}$ where $\widetilde{\Theta }%
(\delta )\equiv \{\widetilde{b}\in \Upsilon :\left\Vert \widetilde{b}-%
\widetilde{\beta }\right\Vert _{E}\leq \delta \}$. For any $s=\left( \tau
,d,z,x\right) $, consider the following function%
\begin{equation}
\overline{m}_{\widetilde{b},\gamma }(z,x)\equiv 1\{z_{1}\beta _{1,1}+%
\widetilde{z}^{\prime }\widetilde{b}{}_{1}+\gamma (x){}^{\prime
}b_{2}>0\}-1\{z_{1}\beta _{1,1}+\widetilde{z}^{\prime }\widetilde{\beta }%
_{1}+\gamma (x){}^{\prime }\beta _{2}>0\}.  \label{m_b_r}
\end{equation}%
Let%
\begin{equation*}
\widetilde{m}_{\widetilde{b},\gamma }(s)\equiv \tau (2d-1)\left[ \overline{m}%
_{\widetilde{b},\gamma }(z,x)-\overline{m}_{\widetilde{b},G}(z,x)\right] .
\end{equation*}%
Define the class of functions%
\begin{equation}
\digamma _{\delta ,\varepsilon }\equiv \left\{ \widetilde{m}_{\widetilde{b}%
,\gamma }:(\widetilde{b},\gamma )\in \widetilde{\Theta }(\delta )\times
H(\varepsilon )\right\} .  \label{F_delta_epsilon}
\end{equation}%
Let $M_{\delta ,\varepsilon }$ denote an envelope for $\digamma _{\delta
,\varepsilon }$ such that $\left\vert \widetilde{m}_{\widetilde{b},\gamma
}(s)\right\vert $ $\leq \left\vert M_{\delta ,\varepsilon }(s)\right\vert $
for all $s$ and for all $\widetilde{m}_{\widetilde{b},\gamma }\in \digamma
_{\delta ,\varepsilon }$.\medskip

\begin{lemma}
\label{envelope copy(1)}Let $\delta $ and $\varepsilon $ be sufficiently
small. Given conditions C1, C4, C6 and C10, for some positive real constants 
$c_{1}$ and $c_{2}$, we can take 
\begin{equation*}
M_{_{\delta ,\varepsilon }}=2\times 1\{c_{1}\min \{\delta ,\varepsilon
\}\geq \left\vert w^{\prime }\beta \right\vert \}
\end{equation*}%
and furthermore, 
\begin{equation}
\left\Vert M_{_{\delta ,\varepsilon }}\right\Vert _{L_{2}(P)}\leq c_{2}\sqrt{%
\min \{\delta ,\varepsilon \}}.  \label{M_delta_epsilon}
\end{equation}
\end{lemma}

\begin{proof}
Note that%
\begin{equation*}
\left\vert \widetilde{m}_{\widetilde{b},\gamma }(s)\right\vert \leq 2\times
1\left\{ \left[ \overline{m}_{\widetilde{b},\gamma }(s)=1\text{ and }%
\overline{m}_{\widetilde{b},G}(s)=-1\right] \text{ or }\left[ \overline{m}_{%
\widetilde{b},\gamma }(s)=-1\text{ and }\overline{m}_{\widetilde{b},G}(s)=1%
\right] \right\} .
\end{equation*}%
Given C1, C6 and C10, there is positive real constant $B$ such that $\max
\{\left\Vert \widetilde{w}\right\Vert _{E},\left\Vert b_{2}\right\Vert
_{E}\}<B$ with probability 1.\ Hence if $(\widetilde{b},\gamma )\in
\widetilde{\Theta }(\delta )\times H(\varepsilon )$, we have that
\begin{eqnarray*}
&&\overline{m}_{\widetilde{b},\gamma }(s)=1 \\
&\iff &z_{1}\beta _{1,1}+\widetilde{z}^{\prime }\widetilde{b}{}_{1}+\gamma
(x){}^{\prime }b_{2}>0\geq z_{1}\beta _{1,1}+\widetilde{z}^{\prime }%
\widetilde{\beta }{}_{1}+\gamma (x){}^{\prime }\beta _{2} \\
&\iff &\widetilde{w}^{\prime }(\widetilde{b}{}-\widetilde{\beta }{})+(\gamma
(x)-G(x)){}^{\prime }b_{2}>-w^{\prime }\beta \geq (\gamma
(x){}-G(x))^{\prime }\beta _{2} \\
&\Longrightarrow &-B(\delta +\varepsilon )\leq w^{\prime }\beta \leq
B\varepsilon .
\end{eqnarray*}%
On the other hand,%
\begin{eqnarray*}
&&\overline{m}_{\widetilde{b},G}(s)=-1 \\
&\iff &z_{1}\beta _{1,1}+\widetilde{w}^{\prime }\widetilde{b}\leq
0<w^{\prime }\beta \\
&\iff &\widetilde{w}^{\prime }(\widetilde{b}{}-\widetilde{\beta }{})\leq
-w^{\prime }\beta <0 \\
&\Longrightarrow &0\leq w^{\prime }\beta \leq B\delta .
\end{eqnarray*}%
Therefore, the condition $\overline{m}_{\widetilde{b},\gamma }(s)=1$ and $%
\overline{m}_{\widetilde{b},G}(s)=-1$ implies $\left\vert w^{\prime }\beta
\right\vert \leq B\min \{\delta $,$\varepsilon \}$. Based on similar
arguments, we can verify that the condition $\overline{m}_{\widetilde{b}%
,\gamma }(s)=-1$ and $\overline{m}_{\widetilde{b},G}(s)=1$ also implies $%
\left\vert w^{\prime }\beta \right\vert \leq B\min \{\delta $,$\varepsilon
\} $. Therefore, Lemma \ref{envelope copy(1)} follows by taking $M_{_{\delta
,\varepsilon }}=2\times 1\{\left\vert w^{\prime }\beta \right\vert \leq
B\min \{\delta ,\varepsilon \}\}$ and noting that given C4, inequality (\ref%
{M_delta_epsilon}) holds for $c_{2}=2\sqrt{2c_{1}L}$.\medskip
\end{proof}

\begin{lemma}
\label{entropy bound copy(1)}Given conditions C1, C4, C6, C7, C8 and C10, we
have that for sufficiently small $\delta $ and $\varepsilon $, 
\begin{equation*}
E\left[ \sup\limits_{\left( b,\gamma \right) \in \widetilde{\Theta }(\delta
)\times H(\varepsilon )}\left\vert \widetilde{S}_{N}(b,\gamma )-\widetilde{S}%
_{N}(b,G)\right\vert \right] \lesssim N^{-1/2}(\max \{\delta ,\varepsilon
\})^{\frac{q}{2\alpha }}(\min \{\delta ,\varepsilon \})^{\frac{\alpha -q}{%
2\alpha }}.
\end{equation*}
\end{lemma}

\begin{proof}
Define the following two classes of functions
\begin{eqnarray*}
A_{\delta ,\varepsilon } &\equiv &\left\{ \tau (2d-1)\overline{m}_{%
\widetilde{b},\gamma }(z,x):(\widetilde{b},\gamma )\in \widetilde{\Theta }%
(\delta )\times H(\varepsilon )\right\} , \\
B_{\delta ,\varepsilon } &\equiv &\left\{ \tau (2d-1)\overline{m}_{%
\widetilde{b},G}(z,x):\widetilde{b}\in \widetilde{\Theta }(\delta )\right\} .
\end{eqnarray*}%
Using Lemma 9.25 of Kosorok (2008, p.\thinspace 169)), we have that
\begin{equation*}
\log N_{[]}(\epsilon ,\digamma _{\delta ,\varepsilon },L_{2}(P))\leq \log
N_{[]}(\epsilon /2,A_{\delta ,\varepsilon },L_{2}(P))+\log N_{[]}(\epsilon
/2,B_{\delta ,\varepsilon },L_{2}(P)).
\end{equation*}%
Let $\psi \equiv \sqrt{\max \{\delta ,\varepsilon \}}$. By Lemma \ref%
{entropy bound}, we have that for sufficiently small $\delta $ and $%
\varepsilon $,
\begin{equation*}
\log N_{[]}(\epsilon ,A_{\delta ,\varepsilon },L_{2}(P))\lesssim \psi
^{2q/\alpha }\epsilon ^{-2q/\alpha }.
\end{equation*}%
Furthermore by simplifying proof of Lemma \ref{entropy bound}, it is
straightforward to verify that
\begin{equation*}
\log N_{[]}(\epsilon ,B_{\delta ,\varepsilon },L_{2}(P))\lesssim \psi
^{2q/\alpha }\epsilon ^{-2q/\alpha }
\end{equation*}%
and thus
\begin{equation}
\log N_{[]}(\epsilon ,\digamma _{\delta ,\varepsilon },L_{2}(P))\lesssim
\psi ^{2q/\alpha }\epsilon ^{-2q/\alpha }.  \label{entropy}
\end{equation}%
Using inequality (\ref{entropy}) and Lemmas \ref{envelope copy(1)}, we have
that
\begin{equation*}
\int\nolimits_{0}^{\left\Vert M_{_{\delta ,\varepsilon }}\right\Vert
_{L_{2}(P)}}\sqrt{\log N_{[]}(\epsilon ,\digamma _{\delta ,\varepsilon
},L_{2}(P))}d\epsilon \lesssim \psi ^{\frac{q}{\alpha }}\int%
\nolimits_{0}^{c_{2}\sqrt{\min \{\delta ,\varepsilon \}}}\epsilon ^{-\frac{q%
}{\alpha }}d\epsilon \lesssim \psi ^{\frac{q}{\alpha }}(\min \{\delta
,\varepsilon \})^{\frac{\alpha -q}{2\alpha }}
\end{equation*}%
where the last inequality follows from the assumption $\alpha >q$. Lemma \ref%
{entropy bound copy(1)} hence follows by applying Corollary 19.35 of van der
Vaart (1998, p.\thinspace 288).\medskip
\end{proof}

We now prove Theorem \ref{thm-distribution}.\medskip

\begin{proof}[Proof of Theorem \protect\ref{thm-distribution}]
By Kim and Pollard\thinspace (1990), $\left\Vert \widehat{\beta }_{G}-\beta
\right\Vert _{E}=O_{p}(N^{-1/3})$. Hence, by condition C11 and Theorem \ref%
{thm-rate}, for sufficiently large $C>0$, probability of the event that $%
\widehat{\beta }\in \Theta _{N}(C),$ $\widehat{\beta }_{G}\in \Theta _{N}(C)$
and $\widehat{G}\in H_{N}(C)$ can be made arbitrarily close to $1$. Thus to
show the theorem, it suffices to establish that for any fixed sufficiently
large $C>0$,
\begin{equation}
\sup\limits_{b\in \Theta _{N}(C)}\left\vert \overline{S}_{N}(b,\widehat{G})-%
\overline{S}_{N}(b,G)\right\vert =o_{p}(N^{-2/3}).
\label{rate of difference of Sn}
\end{equation}%
Given (\ref{rate of difference of Sn}), we have that
\begin{eqnarray*}
\overline{S}_{N}(\widehat{\beta },G) &\geq &\overline{S}_{N}(\widehat{\beta }%
,\widehat{G})-o_{p}(N^{-2/3}) \\
&\geq &\overline{S}_{N}(\widehat{\beta }_{G},\widehat{G})-o_{p}(N^{-2/3}) \\
&\geq &\overline{S}_{N}(\widehat{\beta }_{G},G)-o_{p}(N^{-2/3})
\end{eqnarray*}%
where the first and third inequalities follow from (\ref{rate of difference
of Sn}) and the second inequality follows from the definition of $\widehat{%
\beta }$. Therefore by Theorem 1.1 of Kim and Pollard\thinspace (1990), $%
N^{1/3}(\widehat{\beta }-\beta )$ and $N^{1/3}(\widehat{\beta }_{G}-\beta )$
are asymptotically equivalent in distribution.\medskip

We now verify equation (\ref{rate of difference of Sn}). Given the
decomposition (\ref{decomposition}), it suffices to show that%
\begin{eqnarray}
&&E\left[ \sup\limits_{\left( b,\gamma \right) \in \Theta _{N}(C)\times
H_{N}(C)}\left\vert \widetilde{S}_{N}(b,\gamma )-\widetilde{S}%
_{N}(b,G)\right\vert \right] =o(N^{-2/3}),  \label{e11} \\
&&\sup\limits_{b\in \Theta _{N}(C)}\left\vert \overline{S}(b,\widehat{G})-%
\overline{S}(b,G)\right\vert =o_{p}(N^{-2/3}).  \label{e21}
\end{eqnarray}%
Equation (\ref{e11}) concerns stochastic equicontinuity of the local
recentered process $N^{2/3}\widetilde{S}_{N}(b,\gamma )$ indexed by $\left(
b,\gamma \right) \in \Theta _{N}(C)\times H_{N}(C)$. It is satisfied by
setting $\delta =CN^{-1/3}$ and $\varepsilon =C\varepsilon _{N}$ in the
definition of sets $\Theta \left( \delta \right) $ and $H(\varepsilon )$
given by (\ref{theta_delta}) and (\ref{H_epsilon}) and by invoking Lemma \ref%
{entropy bound copy(1)} with the assumptions $\varepsilon _{N}=o(N^{-1/3})$
and $\alpha >q$.\medskip

We now verify equation (\ref{e21}). Note that for $N$ sufficiently large, if
$b\in \Theta _{N}(C)$, then $b_{1,1}=\beta _{1,1}$ under condition C1. Let
\begin{equation*}
\overline{S}_{1}(\widetilde{b},\gamma )\equiv S_{1}(\widetilde{b},\gamma
)-S_{1}(\widetilde{\beta },\gamma )
\end{equation*}%
where $S_{1}(\widetilde{b},\gamma )$ is defined by (\ref{S1}). Hence it
suffices to verify
\begin{equation*}
\sup\limits_{\widetilde{b}\in \widetilde{\Theta }_{N}(C)}\left\vert
\overline{S}_{1}(\widetilde{b},\widehat{G})-\overline{S}_{1}(\widetilde{b}%
,G)\right\vert =o_{p}(N^{-2/3})
\end{equation*}%
where $\widetilde{\Theta }_{N}(C)\equiv \{\widetilde{b}\in \Upsilon :$ $%
\left\Vert \widetilde{b}-\widetilde{\beta }\right\Vert _{E}\leq CN^{-1/3}\}$.

Note that the term $\left\vert \overline{S}_{1}(\widetilde{b},\widehat{G})-%
\overline{S}_{1}(\widetilde{b},G)\right\vert $ is bounded above by
\begin{equation}
\left\vert S_{1}(\widetilde{\beta },\widehat{G})-S_{1}(\widetilde{\beta }%
,G)\right\vert +\left\vert [S_{1}(\widetilde{b},\widehat{G})-S_{1}(%
\widetilde{\beta },G)]-[S_{1}(\widetilde{b},G)-S_{1}(\widetilde{\beta }%
,G)]\right\vert .  \label{ineq terms}
\end{equation}%
Since $\varepsilon _{N}=o(N^{-1/3})$, by C11 we have that $\left\Vert
\widehat{G}-G\right\Vert _{\infty }=o_{p}(N^{-1/3})$ because the norm $%
\left\Vert \cdot \right\Vert _{\alpha ,p}$ is stronger than the sup norm $%
\left\Vert \cdot \right\Vert _{\infty }$. Hence by Lemma \ref{quadratic
expansion of S(b,r)}, the first term of the sum (\ref{ineq terms}) is $%
o_{p}(N^{-2/3})$ and%
\begin{equation}
S_{1}(\widetilde{b},\widehat{G})-S_{1}\left( \widetilde{\beta },G\right)
=-\left( A_{1}(\widetilde{b})+A_{2}(\widehat{G})+e(\widetilde{b},\widehat{G}%
)\right) +o_{p}\left( N^{-2/3}\right)  \label{t1}
\end{equation}%
where the terms $A_{1}(\widetilde{b})$, $A_{2}(\widehat{G})$ and $e(%
\widetilde{b},\widehat{G})$ are given by (\ref{A1}), (\ref{A2}) and (\ref%
{e(b,r)}), respectively. Using the proof of Lemma \ref{quadratic expansion
of S(b,r)}, it is also straightforward to verify that
\begin{equation}
S_{1}(\widetilde{b},G)-S_{1}(\widetilde{\beta },G)=-A_{1}(\widetilde{b}%
)+o_{p}\left( N^{-2/3}\right) .  \label{t2}
\end{equation}%
Since $\left\Vert \widehat{G}-G\right\Vert _{\infty }=o_{p}(N^{-1/3})$ and $%
\widetilde{b}\in \widetilde{\Theta }_{N}(C)$, we have that $%
A_{2}=o_{p}\left( N^{-2/3}\right) $ and $e(\widetilde{b},\widehat{G}%
)=o_{p}\left( N^{-2/3}\right) $. Putting together (\ref{t1}) and (\ref{t2}),
it follows that the second term of the sum (\ref{ineq terms}) is also $%
o_{p}(N^{-2/3})$ and therefore equation (\ref{e21}) holds.
\end{proof}

\newpage

\begin{center}
\begin{tabular}{cccccc}
\multicolumn{6}{c}{Table 1 : Simulation Results for $\widehat{\beta }_{2,%
\func{Si}ngle}$ and $\widehat{\beta }_{2,OLS}$ (linear $G$)} \\ \hline\hline
$N$ & Bias & RMSE & Median & mean AD & median AD \\ 
\multicolumn{6}{c}{\textit{Single-stage estimation}} \\ 
300 & -0.112 & 0.410 & 0.894 & 0.329 & 0.288 \\ 
500 & -0.061 & 0.352 & 0.916 & 0.285 & 0.247 \\ 
1000 & -0.031 & 0.264 & 0.961 & 0.211 & 0.182 \\ 
\multicolumn{6}{c}{\textit{Two-stage estimation : OLS first stage}} \\ 
300 & -0.122 & 0.478 & 0.856 & 0.381 & 0.326 \\ 
500 & -0.070 & 0.376 & 0.908 & 0.304 & 0.259 \\ 
1000 & -0.033 & 0.301 & 0.952 & 0.240 & 0.211%
\end{tabular}

\bigskip

\bigskip

\bigskip

\bigskip

\bigskip

\bigskip

\begin{tabular}{cccccc}
\multicolumn{6}{c}{Table 2 : Simulation Results for $\widehat{\beta }_{2,%
\func{Si}ngle}$ and $\widehat{\beta }_{2,OLS}$ (nonlinear $G$)} \\ 
\hline\hline
$N$ & Bias & RMSE & Median & mean AD & median AD \\ 
\multicolumn{6}{c}{\textit{Single-stage estimation}} \\ 
300 & -0.056 & 0.330 & 0.918 & 0.262 & 0.216 \\ 
500 & -0.044 & 0.277 & 0.942 & 0.220 & 0.184 \\ 
1000 & -0.020 & 0.212 & 0.966 & 0.169 & 0.139 \\ 
\multicolumn{6}{c}{\textit{Two-stage estimation : OLS first stage}} \\ 
300 & -0.394 & 0.489 & 0.577 & 0.431 & 0.432 \\ 
500 & -0.413 & 0.469 & 0.568 & 0.424 & 0.432 \\ 
1000 & -0.400 & 0.435 & 0.587 & 0.402 & 0.412%
\end{tabular}

\newpage

\begin{tabular}{cccccc}
\multicolumn{6}{c}{Table 3 : Simulation Results for $\widehat{\beta }%
_{2,Kernel\_8th}$ (linear $G$)} \\ \hline\hline
$c$ & Bias & RMSE & Median & mean AD & median AD \\ 
\multicolumn{6}{c}{\textit{Two-stage estimation : kernel first stage }($%
N=300 $)} \\ 
5.4 & -0.063 & 0.639 & 0.868 & 0.502 & 0.417 \\ 
5.6 & 0.041 & 0.653 & 0.966 & 0.500 & 0.400 \\ 
5.8 & 0.138 & 0.718 & 1.067 & 0.544 & 0.427 \\ 
\multicolumn{6}{c}{\textit{Two-stage estimation : kernel first stage }($%
N=500 $)} \\ 
5.4 & -0.046 & 0.501 & 0.908 & 0.388 & 0.314 \\ 
5.6 & 0.056 & 0.518 & 0.992 & 0.393 & 0.307 \\ 
5.8 & 0.171 & 0.584 & 1.096 & 0.435 & 0.331 \\ 
\multicolumn{6}{c}{\textit{Two-stage estimation : kernel first stage }($%
N=1000$)} \\ 
5.4 & -0.087 & 0.389 & 0.887 & 0.311 & 0.266 \\ 
5.6 & 0.008 & 0.380 & 0.992 & 0.307 & 0.264 \\ 
5.8 & 0.111 & 0.424 & 1.086 & 0.334 & 0.278%
\end{tabular}

\bigskip

\bigskip

\bigskip

\bigskip

\begin{tabular}{cccccc}
\multicolumn{6}{c}{Table 4 : Simulation Results for $\widehat{\beta }%
_{2,Kernel\_8th}$ (nonlinear $G$)} \\ \hline\hline
$c$ & Bias & RMSE & Median & mean AD & median AD \\ 
\multicolumn{6}{c}{\textit{Two-stage estimation : kernel first stage }($%
N=300 $)} \\ 
5.8 & -0.071 & 0.480 & 0.918 & 0.380 & 0.328 \\ 
6 & 0.053 & 0.523 & 1.028 & 0.408 & 0.340 \\ 
6.2 & 0.147 & 0.577 & 1.132 & 0.448 & 0.372 \\ 
\multicolumn{6}{c}{\textit{Two-stage estimation : kernel first stage }($%
N=500 $)} \\ 
5.8 & -0.066 & 0.400 & 0.906 & 0.316 & 0.264 \\ 
6 & 0.030 & 0.408 & 1.004 & 0.323 & 0.268 \\ 
6.2 & 0.103 & 0.457 & 1.062 & 0.355 & 0.288 \\ 
\multicolumn{6}{c}{\textit{Two-stage estimation : kernel first stage }($%
N=1000$)} \\ 
5.8 & -0.125 & 0.298 & 0.865 & 0.243 & 0.211 \\ 
6 & -0.028 & 0.286 & 0.956 & 0.230 & 0.192 \\ 
6.2 & 0.059 & 0.320 & 1.038 & 0.251 & 0.211%
\end{tabular}

\newpage

\begin{tabular}{cccccc}
\multicolumn{6}{c}{Table 5 : Simulation Results for $\widehat{\beta }%
_{2,Kernel\_2nd}$ (linear $G$)} \\ \hline\hline
$c$ & Bias & RMSE & Median & mean AD & median AD \\ 
\multicolumn{6}{c}{\textit{Two-stage estimation : kernel first stage }($%
N=300 $)} \\ 
0.6 & -0.172 & 0.483 & 0.803 & 0.392 & 0.345 \\ 
0.8 & -0.122 & 0.502 & 0.865 & 0.395 & 0.328 \\ 
1 & -0.088 & 0.510 & 0.896 & 0.401 & 0.333 \\ 
\multicolumn{6}{c}{\textit{Two-stage estimation : kernel first stage }($%
N=500 $)} \\ 
0.6 & -0.111 & 0.391 & 0.880 & 0.315 & 0.276 \\ 
0.8 & -0.073 & 0.394 & 0.913 & 0.316 & 0.268 \\ 
1 & -0.037 & 0.408 & 0.937 & 0.326 & 0.280 \\ 
\multicolumn{6}{c}{\textit{Two-stage estimation : kernel first stage }($%
N=1000$)} \\ 
0.6 & -0.054 & 0.305 & 0.923 & 0.247 & 0.216 \\ 
0.8 & -0.028 & 0.301 & 0.956 & 0.242 & 0.211 \\ 
1 & 0.002 & 0.313 & 0.980 & 0.250 & 0.216%
\end{tabular}

\bigskip

\bigskip

\bigskip

\bigskip

\begin{tabular}{cccccc}
\multicolumn{6}{c}{Table 6 : Simulation Results for $\widehat{\beta }%
_{2,Kernel\_2nd}$ (nonlinear $G$)} \\ \hline\hline
$c$ & Bias & RMSE & Median & mean AD & median AD \\ 
\multicolumn{6}{c}{\textit{Two-stage estimation : kernel first stage }($%
N=300 $)} \\ 
0.6 & -0.112 & 0.440 & 0.865 & 0.347 & 0.297 \\ 
0.8 & -0.057 & 0.443 & 0.918 & 0.351 & 0.302 \\ 
1 & -0.009 & 0.469 & 0.968 & 0.370 & 0.316 \\ 
\multicolumn{6}{c}{\textit{Two-stage estimation : kernel first stage }($%
N=500 $)} \\ 
0.6 & -0.077 & 0.366 & 0.918 & 0.291 & 0.244 \\ 
0.8 & -0.040 & 0.382 & 0.932 & 0.302 & 0.254 \\ 
1 & -0.010 & 0.397 & 0.966 & 0.313 & 0.264 \\ 
\multicolumn{6}{c}{\textit{Two-stage estimation : kernel first stage }($%
N=1000$)} \\ 
0.6 & -0.037 & 0.272 & 0.952 & 0.218 & 0.182 \\ 
0.8 & -0.012 & 0.272 & 0.980 & 0.218 & 0.192 \\ 
1 & 0.036 & 0.286 & 1.028 & 0.230 & 0.201%
\end{tabular}
\end{center}

\end{document}